\documentclass[12pt,a4paper, doublespacing]{article}
\usepackage[utf8]{inputenc}
\usepackage{graphicx}
\usepackage{geometry}
\usepackage{setspace}
\usepackage{xcolor}
\usepackage{booktabs} 

\usepackage{amsmath, amssymb, amsthm , bbm}
\usepackage{tikz}
\usepackage[title]{appendix}
\usepackage{todonotes}
\usepackage{csquotes}
\usepackage{enumitem}
\usepackage{accents} 
\usepackage{caption}
\usepackage{subcaption}
\usepackage{floatrow}
\newtheorem{theorem}{Theorem}

\newtheorem{definition}[theorem]{Definition}

\newtheorem{remark}[theorem]{Remark}

\newtheorem{condition}[theorem]{Condition}

\newcommand{\R}{\mathbb{R}}

\newcommand{\N}{\mathbb{N}}

\newcommand{\PP}{\mathbb{P}}

\newcommand{\E}{\mathbb{E}}

\newcommand{\cd}{\overset{d}{\longrightarrow}}

\newlength{\dhatheight}

\definecolor{burgundy}{rgb}{0.5, 0.0, 0.13}
\definecolor{darkgreen}{rgb}{0.15, 0.3, 0.1}

\usepackage[left, pagewise, edtable]{lineno}

\usepackage[
backend=biber,
sorting=nyt,
sortcites=true,
natbib=true, 
maxbibnames=10,
minbibnames=3,
mincitenames=1,
maxcitenames=2,
uniquelist=false,
bibstyle=apa, 
citestyle=authoryear, 
giveninits=true, 
uniquename=false,
uniquelist=false, 
isbn=false,
url=false, 
doi=true,
eprint=false,
date=year
]{biblatex}

\bibliography{complete-bibliography.bib}

\usepackage{authblk}

\title{Models for temporal clustering of extreme events with applications to mid-latitude winter cyclones}

\author[1]{Christina Mathieu}
\author[1]{Merle Mendel}
\author[1,2,3]{Katharina Hees}
\author[1]{Roland Fried\thanks{Corresponding author: fried@statistik.tu-dortmund.de}}
\affil[1]{TU Dortmund University, Department of Statistics, 44221 Dortmund, Germany}
\affil[2]{University of Siegen, Mathematics Department, 57068 Siegen, Germany}
\affil[3]{Paul-Ehrlich-Institut, Section of Biometrics, 63225 Langen, Germany}

\begin{document}
	
\maketitle

\doublespacing

	
\begin{abstract}
The occurrence of extreme events like heavy precipitation or storms at a certain location  often shows a clustering behaviour and is thus not described well by a Poisson process. We construct a general model for the inter-exceedance times in between extreme events which combines different candidate models for such behaviour. One of them is formulated in terms of clusters of dependent events with exponential inter-exceedance times in between clusters, while the other assumes independent events separated by heavy-tailed inter-exceedance times. 
We propose a modification of the Cramér-von Mises distance for fitting the combined model. The resulting estimator turns out to be competitive with specialised estimators if the data stem from one of the two submodels. Our modelling approach thus allows us to distinguish these different data generating mechanisms without the need of a-priori model selection. 
An application to mid-latitude winter cyclones illustrates the usefulness of our work as the combination of the two mechanisms improves the  descriptions of such occurrences at many places.   
\end{abstract}


\setlength{\parindent}{0pt}

\section{Introduction}


In metereology and climatology there is large interest in the recurrence times of extreme weather events. For example, mid-latitude cyclones strongly affect the weather conditions, such as temperature, wind, precipitation and cloud cover, and are therefore of great interest (\cite{Dacre2020}). 
Our work is within the peaks-over-threshold framework (see e.g. \cite{Coles2001}) as we model the return times of extreme events with magnitudes which exceed a given threshold. Such extreme events are called \emph{exceedances} and the time between two consecutive exceedances is called \emph{inter-exceedance time} (IET). The exceedance times are traditionally modelled by a Poisson process with i.i.d. exponentially distributed IETs. This can be justified by mathematical arguments: If the event magnitudes are i.i.d. and events are measured in constant time intervals (i.e., equidistant observation times) or in i.i.d. random time intervals following a distribution with existing first moment, the exceedances form a Poisson process asymptotically (e.g., \cite{shanthikumar1983}; \cite{Gut1999}).

In many applications the exceedances show a clustering behaviour with more very short and more very long IETs than expected for a Poisson process. In particular, several studies indicate temporal clustering of mid-latitude cyclones on the west coast of Europe, see \citet{Mailier2006},
\citet{Blender2015}, or
\citet{Dacre2020}. Therefore, we consider two different relaxations of the classical modelling assumption which can describe temporal clustering.

First, if the event magnitudes are not independent but only stationary and a mixing condition that limits long-range dependence is fulfilled, the exceedances form a compound Poisson process asymptotically (\cite{Hsing1988}) where the IETs follow a mixture distribution of the Dirac measure at zero and an exponential distribution (\cite{Ferro2003}). Exceedances then occur in clusters that are asymptotically independent with exponentially distributed recurrence times between subsequent clusters. Approaches to model and estimate temporally clustered extreme events under such assumptions can, e.g., be found in \cite{Fawcett2006,Fawcett2007,Fawcett2012}.
If, in contrast, the event magnitudes are independent but the waiting times between subsequent events are heavy-tailed with infinite mean, then the exceedances form a fractional Poisson process asymptotically (\cite{Laskin2003}; \cite{Meerschaert2011}) with Mittag-Leffler distributed IETs (\cite{Hees2021}). 
\citet{Blender2015} suggest application of this framework for 
modelling mid-latitude cyclones.
Both models for IETs (dependent event magnitudes or heavy-tailed waiting times between subsequent events) describe a temporal clustering behaviour of the extreme events, with short time intervals containing several exceedances followed by long time intervals without any exceedance, but the underlying mechanisms differ widely. 

 \citet{Dissanayake2021} find  the fractional Poisson process to be not flexible enough for modelling the clustering behaviour of significant waves at the Liverpool Bay, UK. The Mittag-Leffler distribution has difficulties describing both the many very short and some very long IETs occurring there. They deduce a need for new methods based on models with a sound mathematical underpinning.    

The goal of this work is to fill this gap and to develop new theory for the behaviour of the IETs when the two conditions, stationary event magnitudes and heavy-tailed waiting times between subsequent events, are met jointly. Such scenarios result in fractional compound Poisson processes with IETs following a mixture distribution with a Mittag-Leffler instead of an exponential component. By considering both mechanisms simultaneously, we get more flexibility for modelling temporal clustering with many very short and some very long IETs, and we do not need to decide in advance which of these mechanisms causes the clustering behaviour. 
In order to make the new model applicable in practice, we need a reliable estimation method for the three model parameters. However, finding such a method is challenging due to the characteristics of our framework and the mixture distribution.
We suggest the minimum distance approach based on a modification of the Cramér-von Mises distance for this task.

The remainder of this work is organized as follows. 
Section \ref{sec:cyclones} introduces the mid-latitude cyclone data on the northern hemisphere that we explore for illustrating our modelling approach. 
Section \ref{sec:probabilisticmodel} introduces probabilistic models that lead to temporal clustering behaviour and combines them to a general model. Section \ref{sec:statinf} discusses some difficulties of finding a suitable estimator for the combined model and suggests a minimum distance estimator based on a modification of the CM distance for this task.
Section \ref{sec:simStudy} evaluates the estimators in a simulation study. In Section~\ref{sec:simStudy2}, we compare the model class introduced in this paper with other clustering models in a further simulation study and discuss their advantages and limitations.
In Section \ref{sec:example}, we apply our modelling approach to the mid-latitude cyclones of Section \ref{sec:cyclones}.
Finally, we close with some conclusions in Section \ref{sec:conclusion}. Proofs and further details are deferred to the Supplementary Material.

\section{Mid-latitude cyclones on the northern hemisphere}
\label{sec:cyclones}

In meteorology, the position of cyclones in the northern hemisphere are typically identified by the maxima of the relative vorticity or the minima of mean sea-level pressure in a given area at a certain time (e.g. \cite{Neu2013}). \citet{Mailier2006} analysed the temporal clustering of mid-latitude cyclones by calculating the variance-to-mean ratio as it measures the degree of deviation from a Poisson point-process (PP) with IETs following an exponential distribution. 

Their results indicate serial clustering on the west coast of Europe, where the exit region of the North Atlantic storm tracks is located, while they occur more regularly in the entry region on the east coast of North America. This pattern has been reproduced in other studies (\cite{Dacre2020}).
\citet{Blender2015} suggest the application of fractional Poisson processes (FPP) to model the clustering behaviour, with IETs following a Mittag-Leffler instead of an exponential distribution. 

We analyse relative vorticity at 850 hPa pressure level of the ERA5 reanalysis data \citep{ERA52023} provided by the European Centre for Medium-Range Weather Forecast (ECMWF) from Winter 1940/41 to 2022/23 including 6h time steps with a horizontal resolution of $1^\circ$ on the North Atlantic Area (30°N - 60°N, 20°E - 80°W). We only use winter data from December, January and February (DJF) due to the different climate and weather dynamics in the other seasons of the year. This is a standard approach in meteorological studies (\cite{Blender2015}; \cite{Neu2013}) and justifies the assumption of (approximately) stationary event magnitudes. 
Each exceedance is associated with the waiting time until the next exceedance, even if it extends beyond the winter period.  This has the advantage that the IETs are not artificially restricted by the end of the season and can last longer than 90 days.
We use the $99\%$ quantile (calculated separately at each grid point) as the respective threshold and consider higher magnitudes as extreme. 

For illustration, we consider three locations on the exit region of the North Atlantic storm track  in detail.
Location A is 3°E 46°N (in the interior of France), Location B 5°E 53°N (west coast of U.K.) and Location C 5°W 52°N (west coast of the Netherlands).
\begin{figure}
\includegraphics[width=1\textwidth]{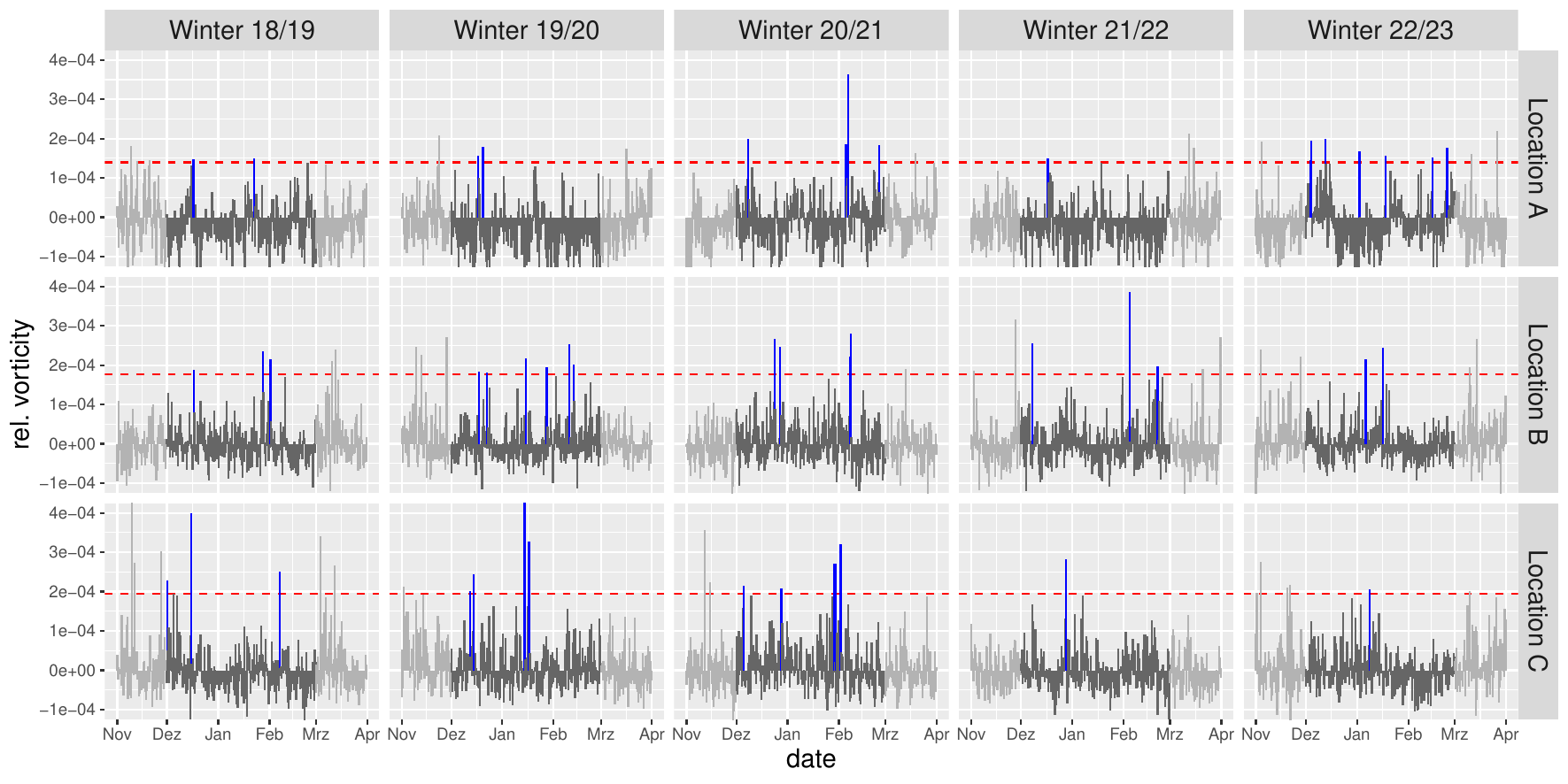}
\caption{Time series of five winters for the three locations.}
\label{fig:exm-5}
\end{figure}
Figure \ref{fig:exm-5} plots the magnitudes of the relative vorticities observed at these three locations for five winters. 
The extreme events crossing the $99\%$ quantile as threshold value show clustering behaviour at all three locations. 

\begin{figure}[h]
     \centering
     \includegraphics[width=\textwidth]{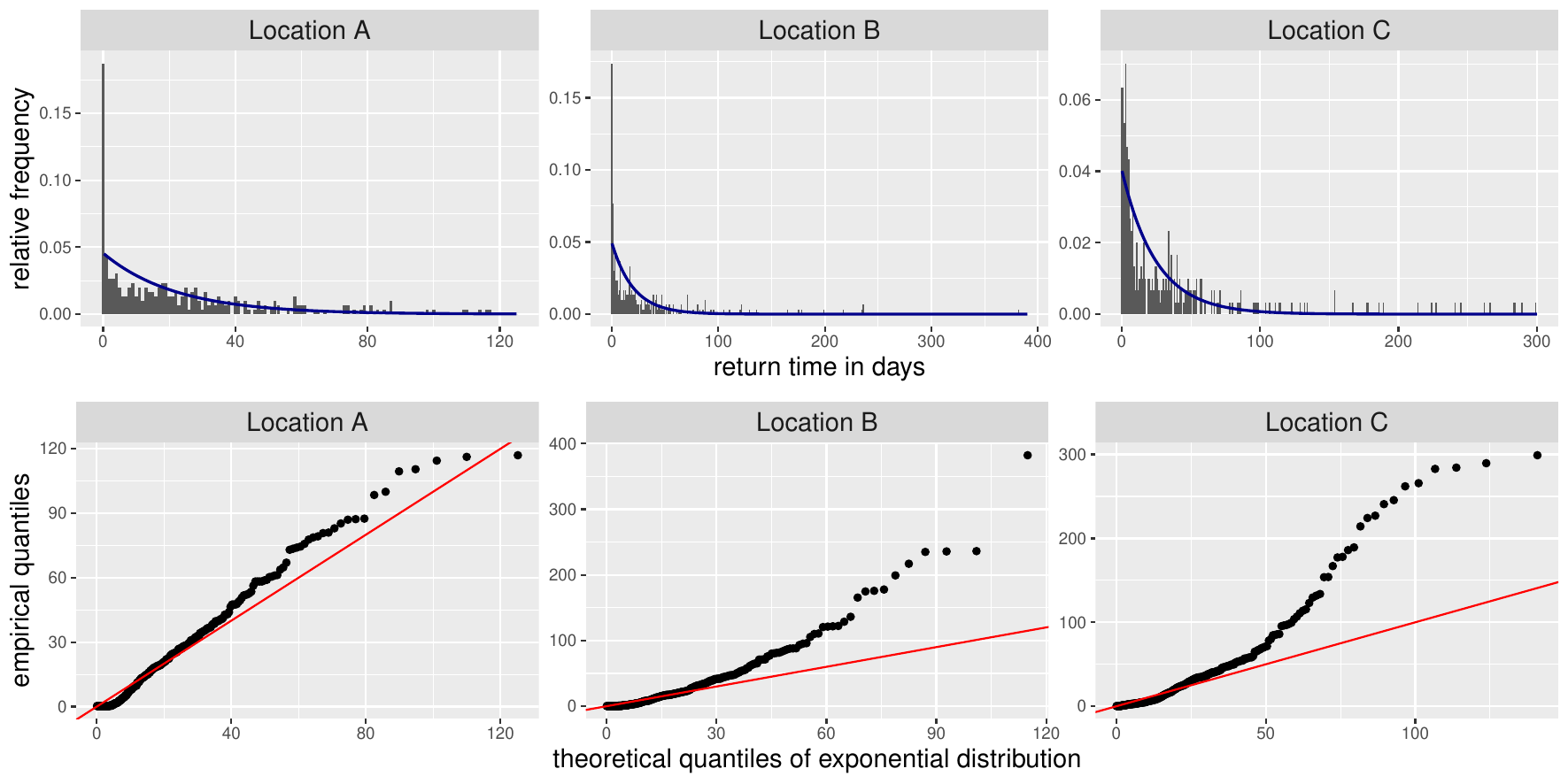}
     \caption{Histogram with bar width of one day (top) and QQ plots (bottom) of the IETs of Location A, $B$ and $C$ with densities and theoretical quantiles, respectively, fitted using the exponential distribution.}
    \label{fig:exm-4}
\end{figure}

Figure \ref{fig:exm-4} shows  histograms  and QQ plots of the IETs for the three locations with densities and theoretical quantiles, respectively, fitted using the exponential distribution. 
We observe that all three locations are poorly described by an exponential distribution. Furthermore, the locations differ from each other. Compared to the exponential distribution, all three locations have a higher probability mass on very small values, although this mass is spread over a larger range at Location C than at Locations A and B. Additionally, the figure suggests that the right tail of the distributions for Locations B and C is heavier than what the exponential distribution would suggest.

In the following we discuss different models for (clustered) exceedances and their IETs which arise from different asymptotic considerations. 

After presenting the mathematical background of these models and the estimation method for the parameters of the mixed distribution in the next three sections, we will return to the analysis of the winter cyclones in Section \ref{sec:example} in order to fit the model to our data.

\section{Asymptotic models for inter-exceedance times}
\label{sec:probabilisticmodel}

Let $(X_n,T_n)_{n \in \N}$ be a marked point process with $T_n$ being the occurrence time and $X_n$ the event magnitude of the $n$-th event. The waiting times between two consecutive events are defined as $W_n = T_n -T_{n-1}$, $n  \in \N$.
We assume that the waiting times $(W_n)_{n \in \N}$ and magnitudes $(X_n)_{n \in \N}$ are independent of each other, and that $X_0 > u$ and $T_0 \equiv 0$ for a given threshold $u$.
Our interest is in the inter-exceedance time of such extreme events:
\begin{align}
	T(u) := T_{\tau(u)} = \sum\limits_{i = 1}^{\tau(u)} W_i ~~~\text{(given $X_0 > u)$}
\end{align}
with stopping time $\tau(u):=\min\{k \ge 1 \mid X_k > u\}$. This means that $X_{\tau(u)}$ is the first magnitude after $X_0$ that exceeds the threshold $u$.

A classical approach to modelling IETs is given by Poisson processes (PP) with i.i.d. exponentially distributed IETs and scale parameter $\sigma>0$, which cannot describe clustering behaviour. Replacing the exponential by the more heavy-tailed Mittag-Leffler distribution with an additional parameter $\beta \in (0,1)$ as suggested by \citet{Blender2015} and \citet{Hees2021} leads to fractional Poisson processes (FPP). The exponential distribution corresponds to the limiting case $\beta \to 1$. Another approach, which is quite popular in the extreme value community, is given by compound Poisson processes (CPP) with IETs following asymptotically a mixture distribution with an exponential component and a point mass $1- \theta$ at 0 after appropriate rescaling \citep{Ferro2003}. 
In the following, we describe the mathematical conditions that lead to these asymptotic results.

Afterwards, we suggest a combination of these approaches called fractional compound Poisson process (FCPP), namely a mixture distribution with a Mittag-Leffler component with parameters $\beta$ and $\sigma$ and a point mass $1-\theta$ at 0. Besides the scale parameter $\sigma$, FCPPs have two further parameters $\theta$ and $\beta$. FPPs, CPPs and PPs correspond to the special cases $\theta=1$, $\beta=1$, or $\theta=\beta=1$.
Using FCPP we have a more flexible model and do not need to decide in advance which submodel fits best.

\subsection*{Background}

For the asymptotic theory, we require that the event magnitudes $X_n$ are identically distributed random variables (r.v.) with the same distribution $\PP^X$ as a r.v. $X$ that belongs to the max-domain of attraction of some non-degenerate distribution $\widetilde{G}$. 

We further consider that the marks $(X_n)_{n \in \N}$ form a (strictly) stationary sequence which fulfills the following mixing condition that limits the long-range dependency:
\begin{condition}
	\label{con:1}
	Let $M(J):=\max \{X_j \mid j \in J \}$  and 
	$\mathcal{I}_{j,l}(u_n):= \{ \{M(I) \le u_n \} \mid I \subset \{j,\dots,l\} \}$. For all $A_1 \in \mathcal{I}_{1,l}(u_n), A_2 \in \mathcal{I}_{l+s,n}(u_n)$ and $1\leq l \leq n-s$,
	$$ |P(A_1 \cap A_2) - P(A_1)P(A_2)| \leq \alpha(n,s) $$
	and $\alpha(n,s_n) \rightarrow 0$ as $n\rightarrow \infty$ for some positive integer sequence $s_n$ such that $s_n=o(n)$. This mixing condition is called $D(u_n)$ condition.
\end{condition}
Condition \ref{con:1} states that two disjoint maxima that are separated by a time lag $s_n$ are approximately stochastically independent as $n \to \infty$. 
We further assume that 
\begin{align}
	\label{eq:3}
	\lim\limits_{n\to \infty} \PP\left( \dfrac{M_n-d_n}{a_n} \le x \right) = G(x) = \widetilde{G}^\theta(x),
\end{align}
for sequences $d_n\in \R$ and $a_n>0$ and
a constant $\theta \in (0,1]$ called \emph{extremal index} of $(X_n)_{n \in \N}$ with $M_n := \max\{X_1,\dots, X_n\}$.
Hereby, $G$ is the c.d.f. of a GEV distribution; see e.g. \citet[chapter 10]{Beirlant2004} for more information on this. 
One can show that for each $\nu \in (0,\infty)$ there is a sequence $u_n$ of thresholds such that
\begin{align}
	& n \cdot \PP(X>u_n) \to \nu ~~~ \text{and} \\
	& \PP(M_n \le u_n) \to \exp(-\theta \nu)
\end{align}
as $n \to \infty$, see e.g. \citet{Leadbetter1983}.

\citet{Ferro2003} derived that
\begin{align}
    \PP(p(u_n)\tau(u_n) > t) \overset{d}{\longrightarrow} \theta \exp(-\theta t) ~~ \text{as } n \to \infty
\end{align}
where $p(u_n) := \PP(X > u_n)$. They used a slightly stronger mixing condition than Condition \ref{con:1}, but as stated in \citet{Beirlant2004}, their result also holds under $D(u_n)$.

In case that the marks $(X_n)_{n \in \N}$ occur in equidistant time intervals, i.e., $T_n = T_{n-1} + 1$ and $W_n \equiv 1 ~~ \forall \, n \in \N$, it holds $T(u) = \tau(u)$ and thus
\begin{align}
	\label{eq:8}
		p(u) T(u) = p(u) \tau(u) \overset{d}{\longrightarrow} T_\theta \text{ as } u \uparrow x_R.
	\end{align}
Hereby, $x_R$ is the right endpoint of the distribution of $X$, and $T_{\theta}$ is a r.v. distributed according to the mixture distribution
\begin{align} 
\PP_\theta := (1-\theta) \cdot \varepsilon_0 + \theta \cdot \text{Exp}(\theta),
\end{align}
with $\varepsilon_0$ being the Dirac measure in $0$ and $\text{Exp}(\theta)$ the exponential distribution with rate $\theta$.
It means that instead of a pure exponential distribution, as is the case for $\theta = 1$, the return times asymptotically follow a mixture distribution with the Dirac measure in zero and the exponential distribution as components.
Thus, the extremal index $\theta$ is related to the times between two exceedances and is responsible for the clustering behaviour. In the limit the IET is either zero, representing the times within a cluster, or exponentially distributed, representing the time between subsequent clusters. Therefore it forms a compound Poisson Process (see e.g. \cite[chapter 10]{Beirlant2004}).

The asymptotics in equation (\ref{eq:8}) can be extended to i.i.d. waiting times $(W_n)_{n \in \N}$ with finite expected value:

\begin{theorem} 
	\label{theorem_FS}
	Assume that the event magnitudes $(X_n)_{n \in \N}$ fulfill assumption (\ref{eq:3}) for some $\theta > 0$ and let the waiting times $(W_n)_{n \in \N}$ be i.i.d. with $\E(W_n) = 1$ for all $n \in \N$. Then
	\begin{align}
	\label{eq:5}
		p(u) T(u) \overset{d}{\longrightarrow} T_\theta \text{ as } u \uparrow x_R,
	\end{align}
	where $p(u):= \PP(X > u)$, $x_R$ is the right endpoint of the distribution of $X$, and $T_{\theta}$ is a r.v. distributed according to the mixture distribution 
    \begin{align} \PP_\theta := (1-\theta) \cdot \varepsilon_0 + \theta \cdot \emph{Exp}(\theta),\end{align}
	with $\varepsilon_0$ being the Dirac measure in $0$ and $\text{Exp}(\theta)$ the exponential distribution with rate $\theta$. 
\end{theorem}

Another mechanism that leads to temporal clustering behaviour are heavy-tailed distributed waiting times, with heavy-tailed meaning that the tail probability is regularly varying with index $-\alpha$ for some $\alpha > 0$. This means that $\PP(W_1 > x)$ $ = x^{-\alpha} L(x)$
equals a power of $x$ up to a 
slowly varying function $L(x)$ which is asymptotically constant, 
\begin{align}
    \lim\limits_{x \to \infty} \frac{L(\lambda x)}{L(x)} = 1 ~~ \text{for all~~} \lambda > 0.
\end{align}
Then the moment $\E(W_1^\gamma) < \infty$ exists if $\gamma < \alpha$, while  $\E(W_1^\gamma) = \infty$ if $\gamma > \alpha$.
A prominent example of a heavy-tailed distribution is the Pareto distribution, which as opposed to the exponential distribution has polynomial tails.

For $\alpha > 1$ the mean of $W_1$ is finite and thus
Theorem \ref{theorem_FS} applies.
For $0 < \alpha < 1$ the waiting time distribution does not have a finite mean. 
In case of an i.i.d. sequence of magnitudes $(X_n)_{n \in \N}$, i.e., extremal index $\theta  = 1$, \citet{Hees2021} showed 
\begin{align} \dfrac{T(u)}{b(1/p(u))} \cd T_{\alpha} \text{ as } u \uparrow x_R \end{align}
under this condition, where $T_{\alpha}$ is a Mittag-Leffler distributed r.v. corresponding to a fractional Poisson Process.

A positive r.v. $T_\beta$ is Mittag-Leffler distributed with parameter $\beta \in (0,1]$ if it has the Laplace transform
\begin{align}
	\mathcal{L}_{}(s)
 := \E(\exp(-s T_\beta))= \frac{1}{1+s^\beta}. 
\end{align}
We write ML$(\beta,\sigma)$ for the distribution of $~\sigma \, T_\beta$, where $\sigma > 0$ is a scale parameter. 
 For $\beta < 1$, the Mittag-Leffler distribution is heavy-tailed with index $\alpha = \beta$ and thus has infinite mean. The exponential distribution is a limiting case, as ML$(1,\sigma) = \text{Exp}(1 / \sigma)$ with mean $\sigma$.
For more information on the Mittag-Leffler distribution see  \citet{Haubold2011}, and for algorithms see the R package \texttt{MittagLeffleR} \citep{Gill2017}.

\subsection*{A general asymptotic model}

Now we bring the two clustering mechanisms together by considering stationary magnitudes $(X_n)_{n \in \N}$ with extremal index $\theta > 0$ and heavy-tailed waiting times $(W_n)_{n \in \N}$ simultaneously.
The following novel theorem states that the limiting distribution of $T(u)$ changes from a mixture distribution with an exponential component or a Mittag-Leffler distribution, respectively, to a mixture with a Mittag-Leffler component. The resulting renewal process changes from a compound Poisson process or a fractional Poisson process, respectively, to a fractional compound Poisson process. See e.g. \citet{Laskin2003} for more information on this model class.

\begin{theorem}
	\label{def:Tu2}
	Assume that the event magnitudes $(X_n)_{n \in \N}$ fulfill  (\ref{eq:3}) for some $\theta > 0$ and let the waiting times $(W_n)_{n \in \N}$ be regularly varying with index $\alpha \in (0,1)$.
	Then, 
	\begin{align}
		\dfrac{T(u)}{b(1/p(u))} \cd T_{\beta,\theta} \text{ as } u \uparrow x_R,
	\end{align}
	where $x_R$ is the right endpoint of the distribution of $X$ and $T_{\beta,\theta}$ is a random variable distributed according to the mixture distribution 
	\begin{align}\PP_{\beta,\theta}:=(1-\theta) \cdot \varepsilon_0 + \theta \cdot \text{ML}(\beta,\theta^{-1/\beta}),\end{align}
	with tail parameter $\beta = \alpha$ and $\varepsilon_0$ being the Dirac measure in 0.
\end{theorem}

\begin{remark}
	In case of an i.i.d. sequence of event magnitudes $(X_n)_{n \in \N}$ the extremal index is $\theta = 1$. Then the return time $T(u)$ is asymptotically Mittag-Leffler distributed according to a fractional Poisson process, see \citet{Hees2021}. 
	
	Waiting times with finite means are covered by the other limiting case $\beta = 1$, with the exponential distribution as a component of the mixture distribution, see Theorem \ref{theorem_FS}. 
\end{remark}

\begin{remark}
    According to Theorem \ref{def:Tu2}, the distribution of the IET $T(u)$ for a given threshold $u$ can be approximated by the asymptotical distribution of $b(1/p(u)) \cdot T_{\beta, \theta}$. For even higher threshold $u + x$, $x > 0$, this implies the approximation $T(u+x) \approx b(1/p(u+x)) \cdot T_{\beta, \theta}$. This is relevant in practice because the threshold $u$ is usually chosen large for the data analysis, but not too large in order to get a sufficient number of observations. With the resulting estimates we can then extrapolate to IETs for even larger, critical thresholds.
\end{remark}

From now on we do not distinguish between the index $\alpha$ and the tail parameter $\beta$. The only exception is the case $\beta = 1$, which corresponds to waiting times with finite means and does not refer necessarily to heavy-tailed waiting times with index $\alpha = 1$.

\begin{remark}
    Since we generally do not know the true distribution of the underlying magnitudes and waiting times, it is difficult to determine the true distribution of $T(u)$ and thus also the convergence rate analytically. Simulations indicate that the convergence rate depends on the underlying distribution and the true parameter values for $\beta$ and $\theta$. Overall, a modified Cramér-von Mises distance (which will be introduced in Section \ref{sec:statinf}) apparently converges with a rate between quadratic and linear. Some illustrations can be found in the supplementary material.
\end{remark}

\section{Model fitting}
\label{sec:statinf}


In this section we treat the estimation of the parameters 
of the mixture distribution derived in Theorem \ref{def:Tu2} using the observed IETs stemming from a sequence of random vectors $(X_i,T_i)_{i=1}^n$ and a threshold $u$.
Restarting the sequence $(X_i,T_i)_{i=1}^n$ at $\tau(u)$, we inductively get the two sequences $(X_j(u))_{j = 0}^{k}$ and $(T_j(u))_{j = 1}^{k}$, where $X_j(u)$ is the $j$-th exceedance of the threshold $u$, and $T_j(u)$ is the IET between $X_{j-1}(u)$ and $X_j(u)$.
Given that we know the previous exceedance, $T_j(u)$ is distributed as $T(u)$ for all $j = 1, \dots, k$.

Theorem \ref{def:Tu2} implies that for a high threshold $u$, we may approximate the distribution of $T(u)$ with the mixture distribution $(1 - \theta) \, \varepsilon_0 + \theta \, \text{ML}(\beta, \theta^{-1/\beta} \sigma_{p(u)})$, where $\sigma_{p(u)} / b(1/p(u))$ is expected to stabilize at a constant as $u$ increases. Thus, in total there are three parameters to be estimated: the tail parameter $\beta$, the extremal index $\theta$ and the scale parameter $\sigma_{p(u)} \approx b(1/p(u)) = p(u)^{-1/\beta} L(1/p(u))$ with a slowly varying function $L$.

The choice of the threshold $u$ means a trade-off between bias and variance: On the one hand, the smaller the threshold, the more exceedances and IETs we have for the estimation (small variance). On the other hand,  the distribution of the IETs may deviate strongly from the mixture distribution (high bias) if the threshold is chosen too low. \citet{Hees2021} explains how to use stability plots for this decision in the situation of fitting a Mittag-Leffler distribution, which corresponds to our special case $\theta=1$. The drawback there is that it is based on a subjective decision and cannot be automated easily. Our focus is not on the choice of the threshold but on the estimation of the parameters from a given sequence of IETs. We do thus not discuss this issue further here.

Searching for a suitable estimation method for estimating $\beta$, $\theta$ and $\sigma_{p(u)}$ simultaneously, we face some difficulties: 
\begin{itemize}
	\item The mixture distribution $\PP_{\beta,\theta, \sigma_{p(u)}}=(1-\theta) \cdot \varepsilon_0+\theta \cdot \text{ML}(\beta, \theta^{-1/\beta} \sigma_{p(u)})$ is neither continuous nor discrete. It is continuous except for a discontinuity point at zero which is the left endpoint of the distribution.
	\item For $\beta < 1$, $\PP_{\beta,\theta,\sigma_{p(u)}}$ is heavy tailed without finite moments.
	\item The observed IETs are all larger than zero, while $\PP_{\beta,\theta,\sigma_{p(u)}}(\{0\}) = 1 - \theta$.
\end{itemize}
These issues make the use of standard estimation methods like maximum likelihood or method of moments difficult or even impossible. 

In this work we propose and investigate \emph{minimum distance estimation} based on  modifications of the \emph{Cramér-von-Mises-distance} for the parameters $\beta \in (0,1]$, $\theta \in (0,1]$ and $\sigma_{p(u)} > 0$ of the mixture distribution.
The minimum distance approach has been introduced by \citet{Wolfowitz1957} and  explored in many further works, see e.g. \citet{Drossos1980} or \citet{Parr1981}.
The main idea is to measure the \enquote{similarity} of the sample data with a parametric model, minimizing a distance measure between the probability density function or the cumulative distribution function of the parametric model and a non-parametric density estimate or the empirical distribution function of the sample data. We use distances based on distribution functions:

\begin{definition}
	Let $Z_1,\dots,Z_n$ be 
    random variables with c.d.f. $F_\vartheta, ~ \vartheta \in \Theta \subset \R^p, ~ p \ge 1$, $F_{n}$ the empirical c.d.f. corresponding to $Z_1,\dots,Z_n$, and $\Delta(\cdot,\cdot)>0$ a function quantifying the distance between two c.d.f.'s. If there is a $\hat \vartheta \in \Theta$ such that
	\begin{align}
		\Delta(F_n,F_{\hat \vartheta}) = \inf\limits_{\vartheta \in \Theta} \Delta(F_n,F_\vartheta),
	\end{align}
	then $\hat \vartheta$ is called a \emph{minimum distance estimate} of $\vartheta$.
\end{definition}
Here, $\Delta(\cdot, \cdot)$ is called \emph{criterion function}. We use a modification of the popular Cramér-von-Mises (CM) distance as criterion function. 

The Cramér-von-Mises distance between two c.d.f.'s $G$ and $H$
is defined as
$\Delta^{[\text{\tiny CM}]}(G,H) = \int_{- \infty}^{\infty} (G(x)-H(x))^2 \, \text{d}H(x)$.
Let $F_{\beta,\theta,\sigma_{p(u)}}$ be the c.d.f. of $\PP_{\beta,\theta, \sigma_{p(u)}}$ and $F^*_{\beta,\theta, \sigma_{p(u)}}$ the c.d.f. of the Mittag-Leffler distribution $\text{ML}(\beta,\theta^{-1/\beta} \sigma_{p(u)})$. 
Because of
$F_{\beta,\theta,\sigma_{p(u)}}(x) = (1-\theta) \cdot \mathbbm{1}_{[0,\infty)}(x)+\theta \cdot F^*_{\beta,\theta, \sigma_{p(u)}}(x)$,
the Cramér-von-Mises distance between $F_{\beta,\theta, \sigma_{p(u)}}$ and the empirical c.d.f. $F_{k}$ of the $k$ observed IETs $t_1, \dots, t_k$ is
\begin{align}
	\Delta^{[\text{\tiny CM}]}(F_{k} , F_{\beta,\theta,\sigma_{p(u)}})
	& = (1-\theta)^3 + \theta \cdot \int_{0}^{\infty} (F_{k}(x) - F_{\beta,\theta,\sigma_{p(u)}}(x))^2 \, \text{d}F^*_{\beta,\theta,\sigma_{p(u)}}(x).
\end{align}
The smaller the value of $\theta$ is, the less influence have the data on the distance $\Delta^{[\text{\tiny CM}]}$. 
Irrespective of  the underlying true parameter values it holds that $\Delta^{[\text{\tiny CM}]}(F_{k} , F_{\beta,\theta,\sigma_{p(u)}}) >(1-\theta)^3$ and $\lim\limits_{\theta \to 0}$ $\Delta^{[\text{\tiny CM}]}(F_{k} , F_{\beta,\theta,\sigma_{p(u)}})  = 1$. 
Since $\Delta^{[\text{\tiny CM}]}(F_{k} , F_{\beta,\theta,\sigma_{p(u)}})  \in [0,1]$, this can lead to a huge bias when we search for the infimum of $\Delta^{[\text{\tiny CM}]}(F_{k} , F_{\beta,\theta,\sigma_{p(u)}})$.

We consider the following modification CMmod of the Cramér-von-Mises distance:
\begin{align}
	\Delta^{[\text{\tiny CMmod}]}(\tilde{F}_{k} , F_{\beta,\theta,\sigma_{p(u)}})
	& = \frac{1}{\theta^2} \int_{0}^{\infty} (\max\{\tilde{F}_{k}(x),1-\theta\} - F_{\beta,\theta,\sigma_{p(u)}}(x))^2 \, \text{d}F^*_{\beta,\theta,\sigma_{p(u)}}(x),
\end{align}
where $\tilde{F_k}$ is the empirical c.d.f. of $t_1 + 1, t_2 + 1, \dots, t_k + 1$, the by one shifted observed IETs. Some explanations are given in Remark \ref{remark:1} below.
\begin{remark}
	\label{remark:1}
	\begin{enumerate}
		\item[]
        \item We get \emph{CMmod} by only considering the continuous part of the integrator of \emph{CM}.
        
		\item We truncate the empirical c.d.f., because $F_{\beta,\theta,\sigma_{p(u)}}(x) > 1-\theta$ for all $x>0$. 
        \item Since $(\max\{\tilde{F}_{k}(x),1-\theta\} - F_{\beta,\theta,\sigma_{p(u)}}(x))^2 \in [0,\theta^2]$, we additionally standardise it with $\theta^2$.
        \item We use $t_i + 1$ instead of $t_i$ for all $i \in \{1, \ldots, k\}$, because prior simulations have shown that this improves parameter estimation and the asymptotics from Theorem \ref{def:Tu2} still hold since $b(1/p(u)) \to \infty$:
        \begin{align}\frac{T(u) + 1}{b(1/p(u))} \cd T_{\beta, \theta}.\end{align}
	\end{enumerate}
\end{remark}
For computations we prefer rewriting the distances in terms of sums. 
After some cumbersome but straightforward calculations we get
\begin{align}
\begin{array}{ll}
	& \Delta^{[\text{\tiny CMmod}]}(\tilde{F}_{k} , F_{\beta,\theta,\sigma_{p(u)}}) \\
	& = \frac{1}{\theta^3} \frac{1}{k} \sum_{i = l + 1}^{k} \left(\frac{i-\frac{1}{2}}{k}-F_{\beta,\theta,\sigma_{p(u)}}(t_{(i)} + 1)\right)^2
	+ \frac{k - l}{12 k^3 \theta^3}
	- \frac{\left(k(1-\theta)\right)^3 - l^3}{3 k^3 \theta^3} \nonumber \\
	& ~~~~
	+ \frac{\left(k(1-\theta)\right)^2 - l^2}{k^2 \theta^3} F_{\beta,\theta,\sigma_{p(u)}}\left(t_{(l)} + 1\right) \nonumber
	- \frac{k(1-\theta) - l}{k \theta^3} F_{\beta,\theta,\sigma_{p(u)}}\left(t_{(l)} + 1\right)^2,
    \end{array}
\end{align}
where $t_{(1)} < \dots < t_{(k)}$ are the ordered IETs and $l := \lceil k(1-\theta) \rceil$, with $\lceil \cdot \rceil$ being the ceiling function and $k$ the number of IETs.

The CMmod distance converges to $1 / 3$ as $\theta \to 0$, since for $\theta < 1 / k$, $l = \lceil k (1 - \theta) \rceil = k$ and therefore for $1/k >  \theta$, it holds
$\Delta^{[\text{\tiny CMmod}]}(\tilde{F}_{k} , F_{\beta,\theta,\sigma_{p(u)}}) 
    = \frac{1}{3}$.
We thus suggest restricting the parameter spaces of both, $\beta$ and $\theta$, to a compact interval $[a,1]$ for some lower bound $a > 1 / k$, so that the minimum distance estimate $(\hat\beta,\hat\theta, \hat\sigma_{p(u)})$ of $(\beta,\theta,\sigma_{p(u)})$ shall fulfill
\begin{align}
	& \Delta^{[\text{\tiny CMmod}]}(\tilde{F}_{k},F_{\hat\beta,\hat\theta,\hat\sigma_{p(u)}}) = \inf\limits_{\substack{\beta,\theta \in [a,1] \\ \sigma_{p(u)} \in (0, \infty)}} \Delta^{[\text{\tiny CMmod}]}(\tilde{F}_{k},F_{\beta,\theta,\sigma_{p(u)}}).
\end{align}
The lower boundary $a$ can be chosen depending on the situation and prior knowledge. We believe that $a = 0.1$ might usually be an appropriate choice, since we expect that the true parameter value is usually larger than this. Otherwise about $90\%$ of the inter-exceedance times would be close to zero.

We also explored further modifications of the CM distance. 
However, they turned out to be less suitable and are thus not considered here.

\section{Simulation Study}
\label{sec:simStudy}


In this section we analyse the performance of the minimum distance method proposed above.
All statistical computations are done with R \citep{RCore}. 

\subsection*{Scenarios}
We generate 1000 event sequences for each of several scenarios. We consider event sequences from  max-autoregressive processes defined as
\begin{align}
	X_1 & := Y_1 \\
	X_{i+1} & := \max\{(1-\theta) \cdot X_{i} , \theta \cdot Y_{i+1}\},
\end{align}
where $Y_i$, $i=1,\dots,n$, are independent unit Fréchet random variables and $\theta \in \{0.5,0.6,\dots,1\}$ is the extremal index. 
In case of $\beta=1$, we consider the following distributions for the stochastically independent waiting times $W_i$, $i=1,\dots,n$:
\begin{enumerate}[label = (\alph*)] 
	\item Exponential distribution with mean equal to one.
	\item Dirac measure at point one (i.e., deterministic waiting times equal to one).
	\item Pareto distribution with stability parameter $\alpha = 1.5$ and mean equal to one but infinite variance.
    \item Pareto distribution with stability parameter $\alpha = 2.5$, mean equal to one and finite variance. 
\end{enumerate}
For $\beta < 1$ the waiting times are in the domain of a positively skewed sum-stable distribution with stability parameter $\beta \in \{0.5, \dots,0.9\}$. We consider these three distributions: 
\begin{enumerate}[label = (\roman*)] 
	\item stable distribution,
	\item Mittag-Leffler distribution and
	\item Pareto distribution with shift one,
\end{enumerate}
such that the slowly varying component $L(n)$ of $b(n) = n^{1/\beta} L(n)$ is constant equal to one. Therefore, we consider $\rho = \sigma_u \cdot p(u)^{1/\beta} \approx 1$ instead of $\sigma_u \approx b(1/p(u)) = p(u)^{-1/\beta}$ as scaling parameter. More details regarding the waiting time distributions can be found in the supplementary material. For illustration, Figure~\ref{fig:sim2} shows the densities of the waiting time distributions presented above.

\begin{figure}[htb] 
	\centering
	\includegraphics[width = 1\textwidth]{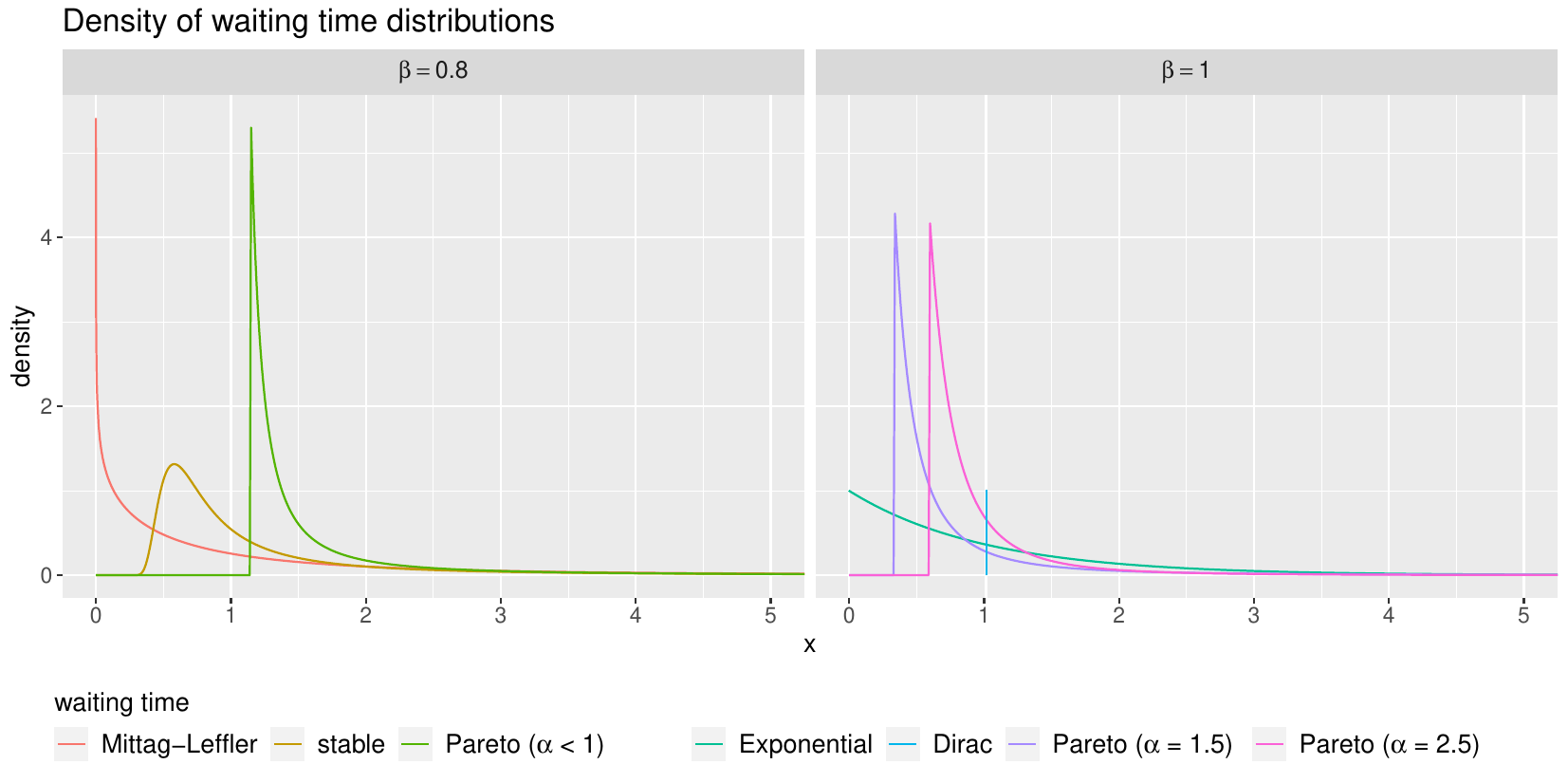}
	\caption{Density of the continuous waiting time distributions with tail parameter $\beta = 0.8$ (left) or tail parameter $\beta = 1$ (right). In case of the Dirac measure it is the probability mass function.}
	\label{fig:sim2}
\end{figure}

We focus on sequences $(X_i, W_i)_{i = 1, \dots, n}$ of length $n = 10000$. In our context, this means that if we had on average hourly observations, we would need data from about 14 months to reach $n = 10000$ observations; if we had daily observations, we would need data from about 27.4 years, and if we had on average weekly observations, we would need data from about 185 years. Results for other sample sizes are given in the supplementary material.

We set the threshold to the 98\% sample quantile, such that the 2\% largest magnitudes are considered as exceedances. For $n = 10000$ observations this leads to $k = 200$ exceedances. In general, selecting an appropriate threshold is a difficult task. It cannot be set too high because we require a sufficient number of inter-exceedance times to compute the empirical distribution function. Conversely, the approximation may not be accurate if the threshold is set too low. Previous studies not reported here suggest that 2\% is a reasonable choice in our scenarios. 

For minimisation we use the standard optimisation algorithm \emph{L-BFGS-B} based on quasi-Newton with several starting points \citep{Byrd1995}. We restrict the search space to $[a,1]\times[a,1]\times(0,\infty)$ with $a = 0.1$ as discussed before. We report the root of the mean-square error (RMSE) and the bias of the point estimators.

\subsection*{Results}
When reporting the simulation results, we pay attention to the differences between the waiting time distributions. In the special cases $\beta = 1$ and $\theta = 1$ we compare our estimators with established estimators for these scenarios.

Overall, the results of the simulation study in Figure~\ref{fig:sim1} are rather satisfactory and differ only slightly with respect to the different waiting time distributions in general. In some cases, the Pareto distribution leads to a slightly larger bias. In almost all cases, the bias and RMSE decrease as the parameter values for $\beta$ and $\theta$ approach their upper limit 1. 

\begin{figure}[p] 
	\centering
	\includegraphics[width = 1\textwidth]{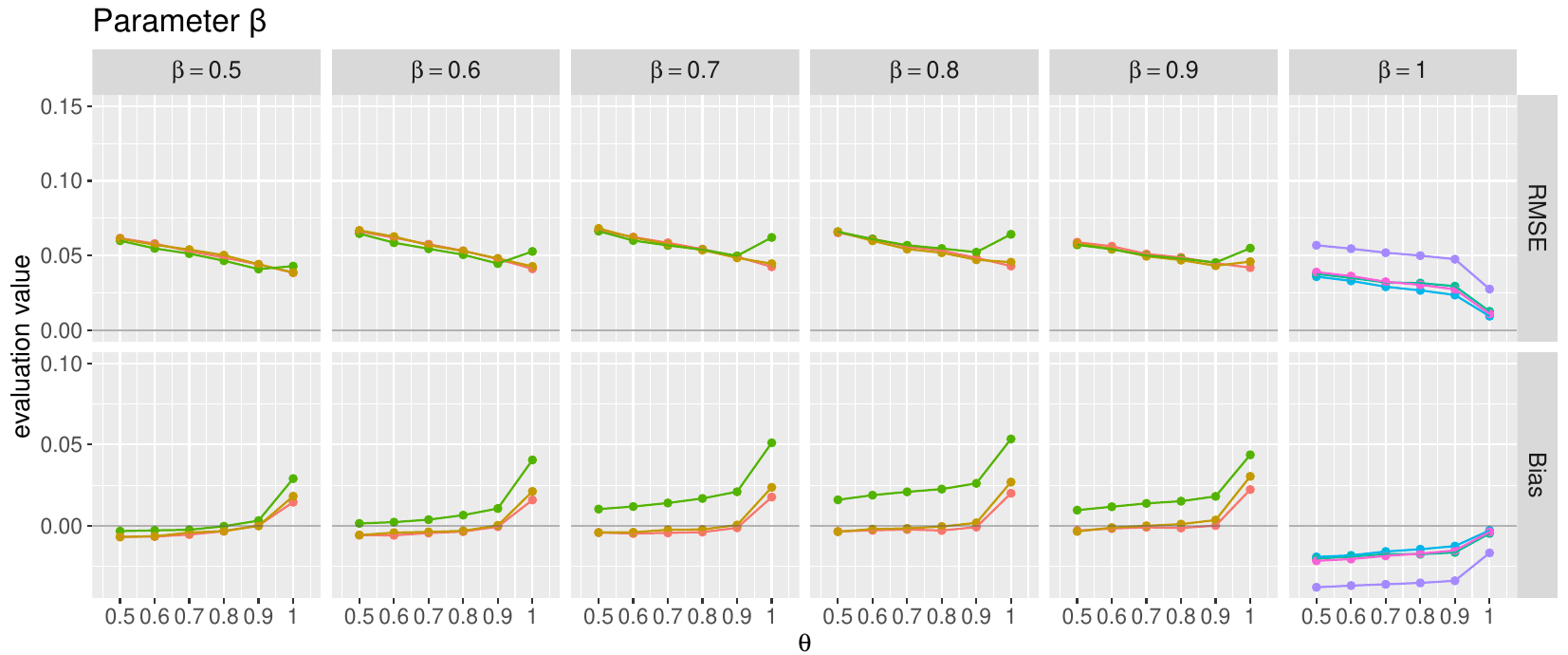}
 	\includegraphics[width = 1\textwidth]{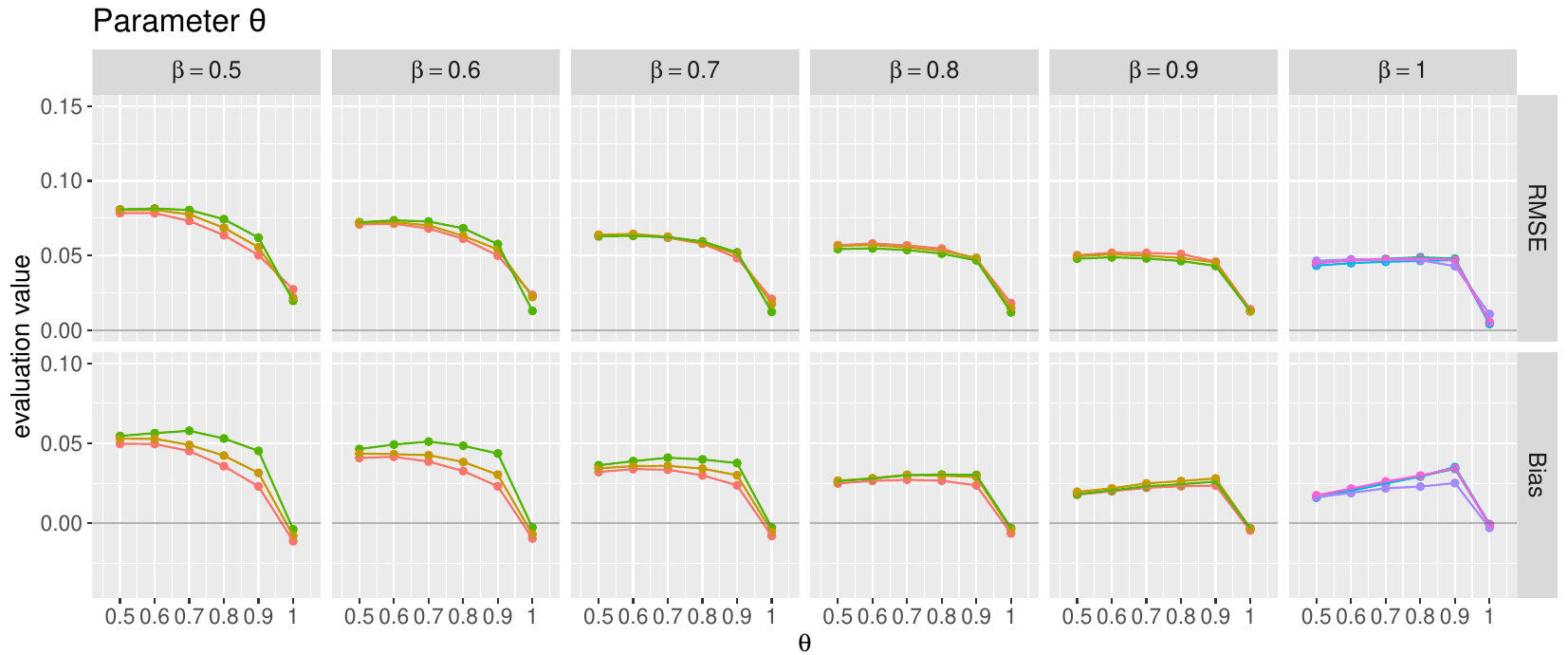}
  	\includegraphics[width = 1\textwidth]{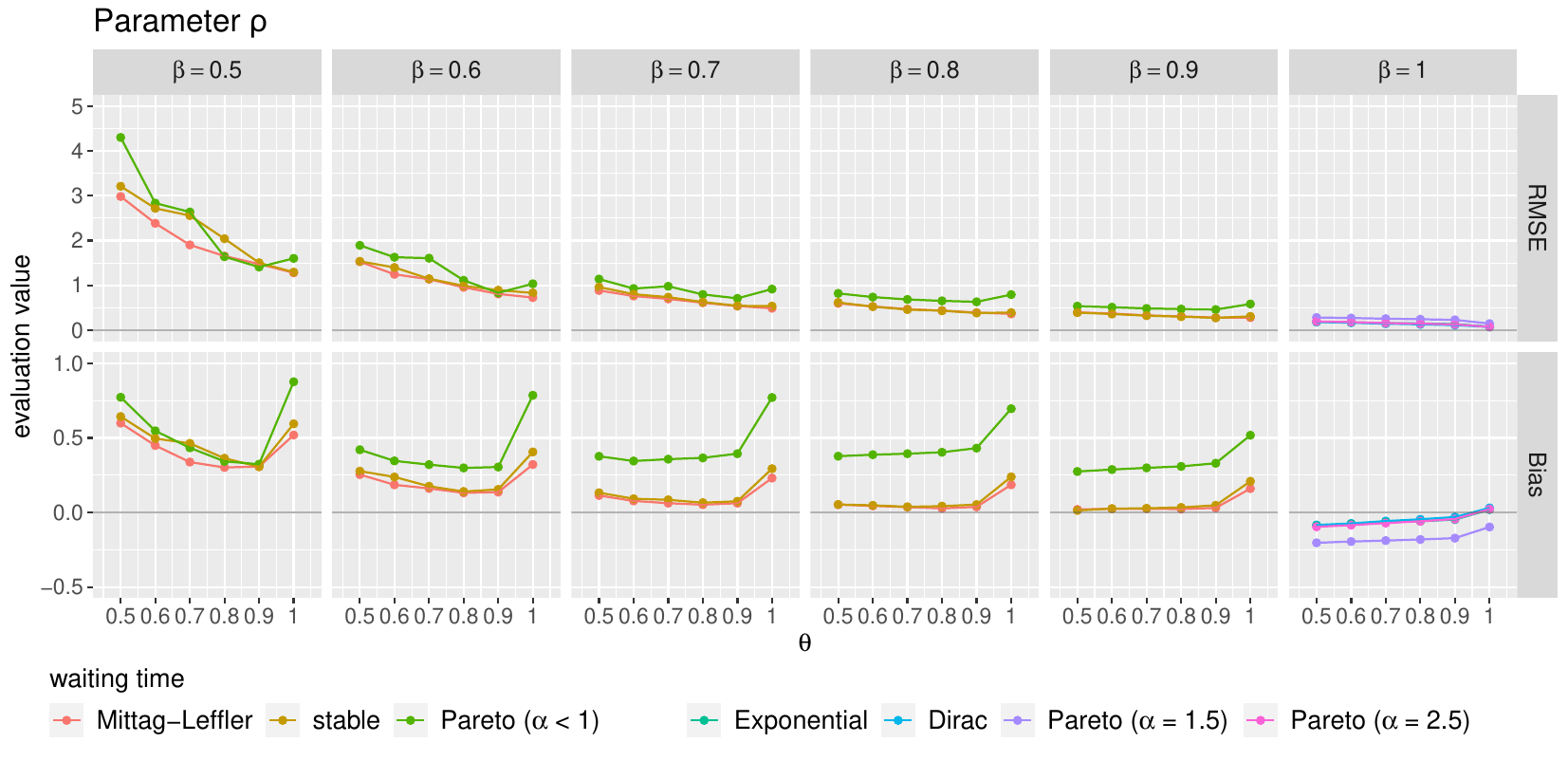}
	\caption{RMSE and Bias of the CMmod estimator for the tail parameter $\beta$ (top), for the extremal index $\theta$ (middle) and the scale parameter $\rho$ (bottom).}
	\label{fig:sim1}
\end{figure}

\begin{figure}[p] 
	\centering
	\includegraphics[width = 1\textwidth]{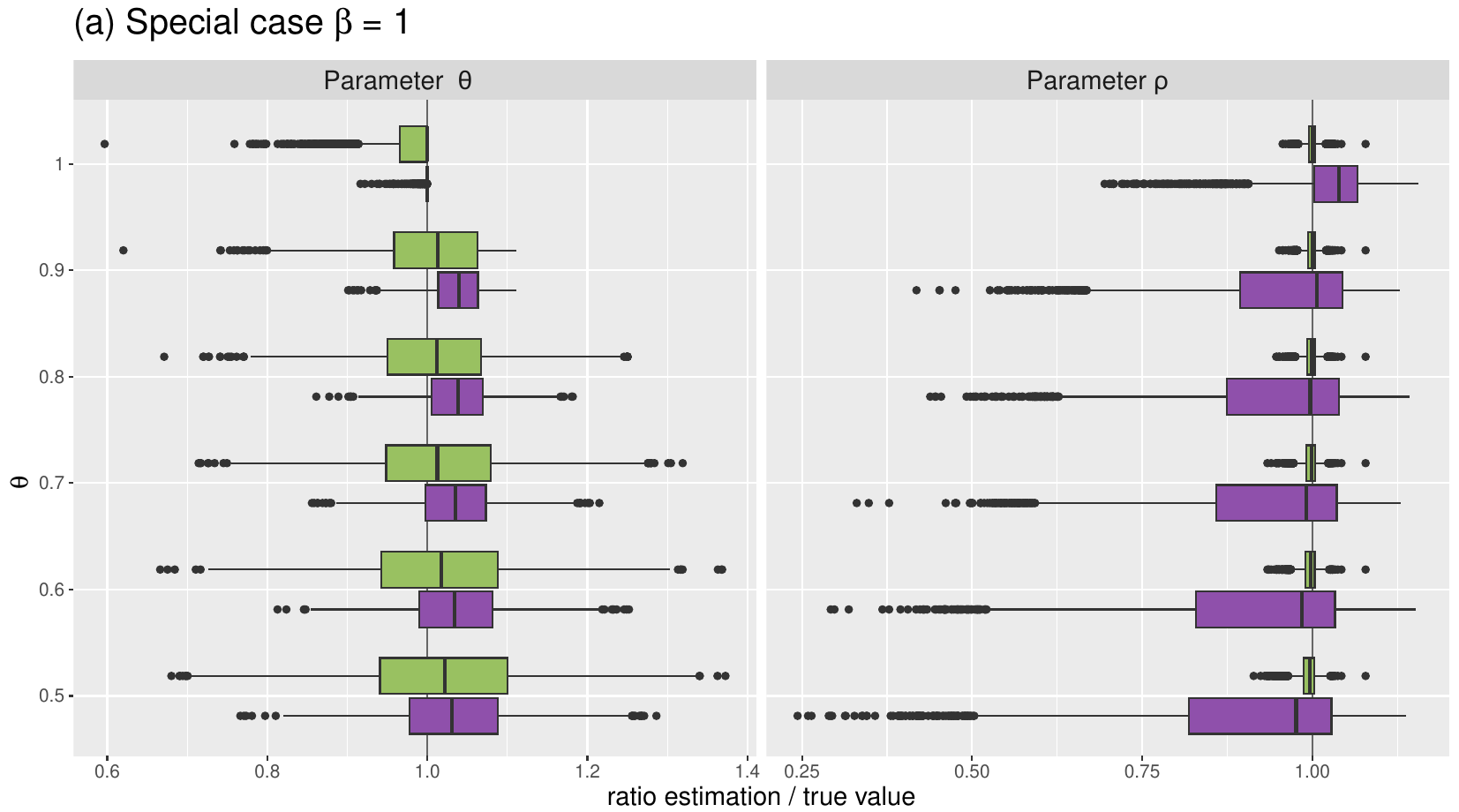}
 	\includegraphics[width = 1\textwidth]{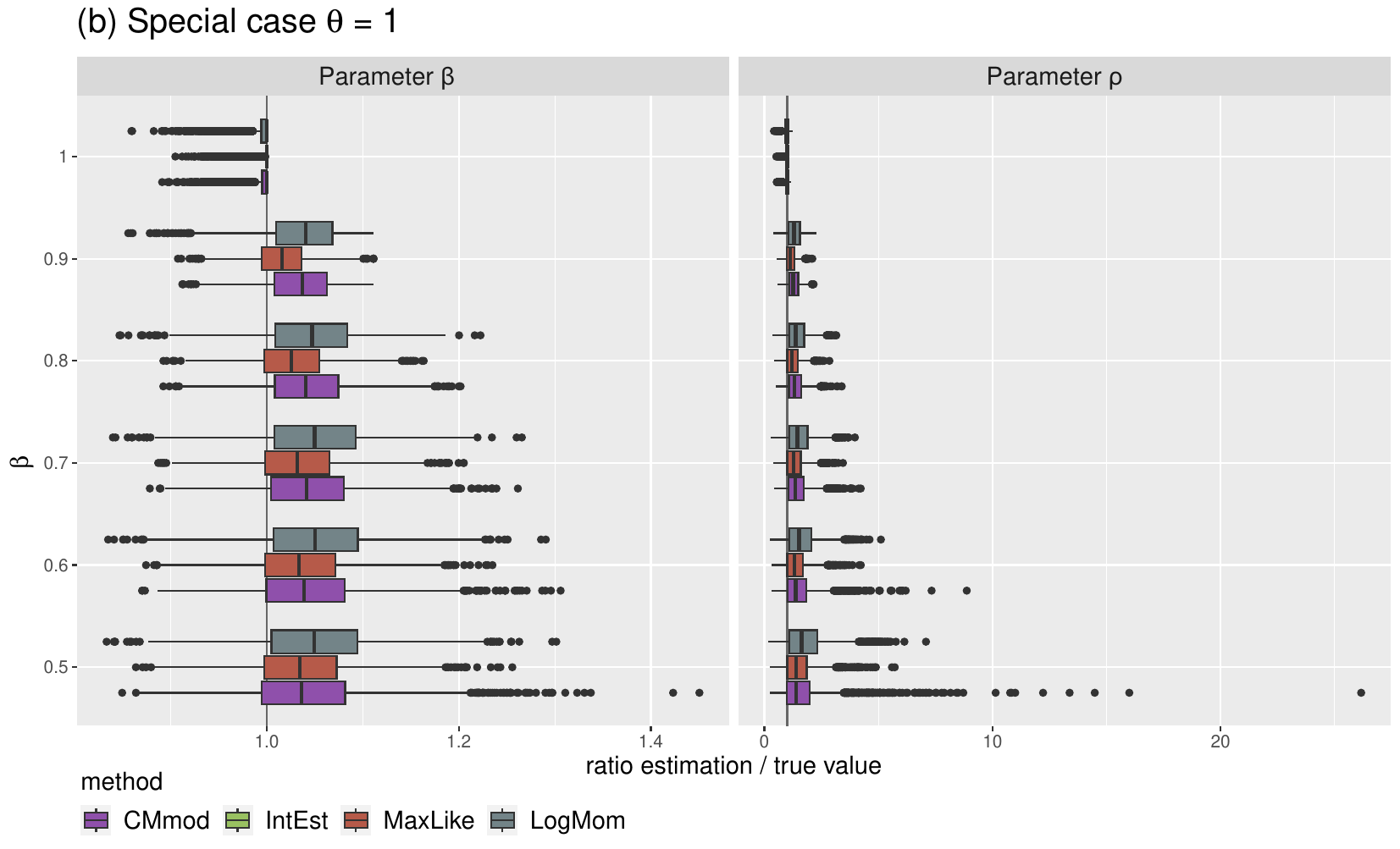}
	\caption{Comparison of CMmod with (a) the interval estimator in the special case $\beta = 1$ across all waiting time distributions except the Pareto distribution with $\alpha = 1.5$ (top), and with (b) the maximum likelihood and log-moment estimator in the special case $\theta = 1$ across all waiting time distributions.}
	\label{fig:sim3}
\end{figure}

As discussed in Sections \ref{sec:probabilisticmodel} and \ref{sec:statinf}, for $\beta = 1$ we are in the special case of a mixture distribution with the exponential distribution as continuous part, $T(u) \approx (1-\theta) \varepsilon_0 + \text{Exp}(p(u) \rho^{-1})$ with $\varepsilon_0$ being the Dirac measure. The IETs are not heavy-tailed then and $\theta$ is called \emph{extremal index}, see e.g. \citet{Beirlant2004}. Among the many estimators of the extremal index we choose the popular interval estimator $\hat{\theta}_I$ of \citet{Ferro2003} for comparison, since it uses the IETs and does not need any hyperparameters for calculation. We need to adapt it slightly since we may have IETs $T_1(u), \ldots, T_k(u)$ smaller than one. We therefore replace $T_i(u)-1$ with $\max\{T_i(u) - 1, 0\}$ and $T_i(u)-2$ with $\max\{T_i(u) - 2, 0\}$, respectively.

The mean of the exponential distribution $p(u)^{-1} \rho$ can be estimated separately by the mean of the waiting times $W_1, \dots, W_n$ multiplied with $n/k$, or by the mean of the IETs $T_1(u), \dots, T_k(u)$. 

Figure~\ref{fig:sim3} (a) shows the results of the interval estimator and the CMmod estimator where the boxplots are calculated across all waiting times, excluding the Pareto distribution with $\alpha = 1.5$ (details below). The minimum distance method shows a slightly larger bias for the extremal index $\theta$, but its variability is typically smaller resulting in a smaller RMSE. The scale parameter $\rho$ is estimated more accurately by the interval estimator. However, the interval estimator struggles when the waiting times are Pareto distributed with stability parameter $\alpha = 1.5$, and its estimation accuracy does not improve for larger sample sizes. This is plausible, since the variance does not exist and the interval estimator uses the ratio of the first two distribution moments. Therefore, this distribution is not included in the results of Figure~\ref{fig:sim3} (a). See the supplementary material for the comparison in case of the Pareto distribution with stability parameter $\alpha = 1.5$.

 For $\theta = 1$ we are in the special case of asymptotically Mittag-Leffler distributed IETs, i.e., $T(u) \approx ML(\beta, p(u)^{-1/\beta} \rho)$. Thus we can compare our estimation method for $\beta$ and $\rho$ with the established maximum likelihood estimator and the log-moment estimator \citep{Cahoy2010} for the tail and scale parameter of the Mittag-Leffler distribution, which are based on the log-transformed data. Both are implemented in the R package \texttt{MittagLeffleR}. 
The comparison for the tail and the scale parameter is shown in Figure~\ref{fig:sim3} (b). Maximum likelihood usually shows the best performance. This confirms findings by \citet{Hees2021} that the maximum likelihood estimator often outperforms the log-moment estimator. The CMmod estimator performs similarly to the log-moment estimator, although it needs to estimate the parameter $\theta$ additionally. For larger sample sizes, the results of the CMmod estimator are even better than those of the log-moment estimator and similar to the results of the maximum likelihood estimator (see illustrations in the supplementary material).

Overall, CMmod shows quite satisfactory performance even in both special cases, although it needs to estimate one parameter more than the competitors which are designed for these scenarios.
A drawback is the high computing time of the minimum distance method. Numerical optimisation is needed to find the triplet $(\hat{\beta},\hat{\theta},\hat{\sigma}_u)$ for which the distance is minimal. Because of possible multiple local minima we use several initialisations ($\{0.25,0.55,0.85\}^2 \times \{\hat{\sigma}_{\text{LogMom}}\}$), where $\hat{\sigma}_{\text{LogMom}}$ is the log-moment estimator in case of the Mittag-Leffler distribution. The computing time seems to be linear in the number of exceedances $k$ (see Figure~\ref{fig:sim4}). 
The computing times needed for the maximum likelihood method in the special case $\theta = 1$ are shown for comparison. They are also much higher than those of the log-moment and the interval estimator.

\begin{figure}[htb]
	\centering
	\includegraphics[width = 0.66\textwidth]{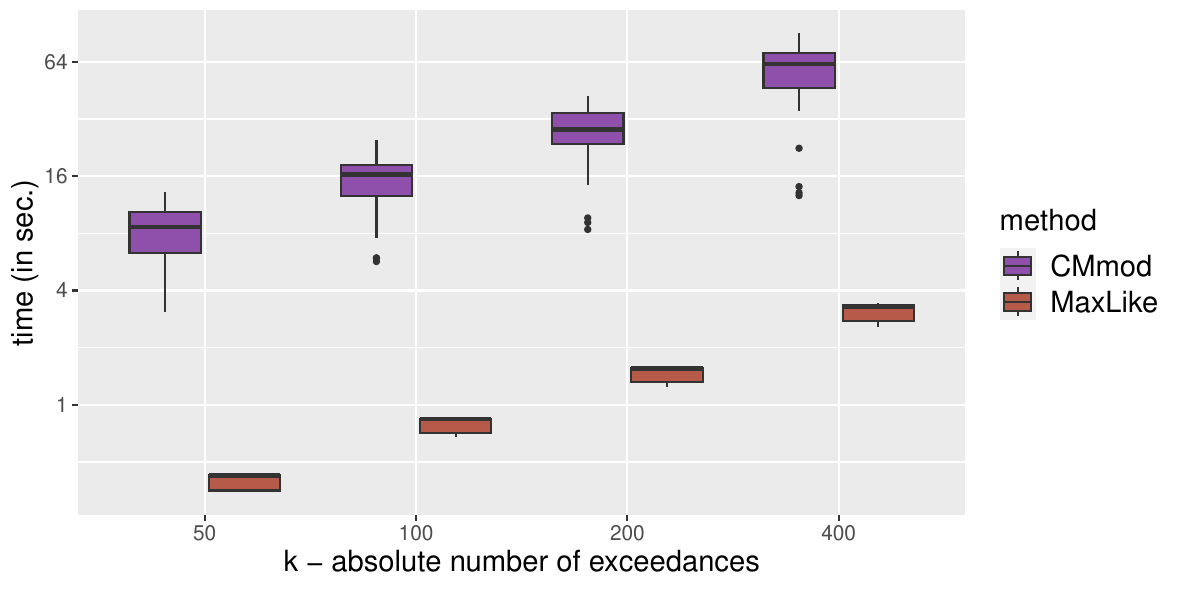}
	\caption{Computing time of CMmod and the Mittag-Leffler maximum likelihood estimator for an increasing number of exceedances $k = 50, 100, 200, 400$. Note that both the x-axis and y-axis are log-transformed.}
	\label{fig:sim4}
\end{figure}

\subsection*{Parametric bootstrap procedures}

In the upcoming data analysis in Section \ref{sec:example}, we use parametric bootstrap methods for statistical inference. We now examine these in a small simulation study to demonstrate their reliability.
Due to the high computation time required, we only consider waiting times with stable distributions. For each parameter selection, we choose $N=100$ simulation runs of size $n=10000$ with $k=200$ exceedances. For each simulation run with parameter estimates $(\hat{\theta},\hat{\beta},\hat{\rho})$, we then generate $B=100$ bootstrap samples of the same size from a FCPP with these parameter values and compare the $N$ estimates obtained from the bootstrap datasets  to the estimate obtained jointly from the $N$ original datasets. 
For a significance test of the null hypothesis $\beta=1$ ($\theta=1$) at a given significance level such as $\alpha=0.05$ we reject this hypothesis if less than $100 \cdot \alpha\%$ of the $B$ bootstrap estimates for this parameter equal one.

Figure~\ref{fig:sim5} displays the estimated rejection rates  of the bootstrap tests for the two null hypotheses $H_0:\beta=1$ and $H_0:\theta = 1$. 
In case of the test for the hypothesis $\beta=1$, the empirical rejection rates are below 9\% under the null hypothesis in all scenarios considered here. The rejection rate increases quickly as $\beta$ decreases. 
In case of the test for the hypothesis $\theta=1$, a more stringent significance level of 3\% is maintained in the scenarios considered here. The rejection rate increases moderately at first, but for $\theta \le 0.8$, the test rejects the null hypothesis reliably.

These results highlight the trustworthiness of the classification obtained in Section \ref{sec:example} for the real data (see Figure~\ref{fig:exm-3}).
In a similar manner, approximate two-sided confidence intervals could be calculated by using the $\alpha/2$ and the $1-\alpha/2$ quantile of the estimates obtained for the bootstrap samples as boundaries for each parameter.

\begin{figure}[htb]
	\centering
	\includegraphics[width = 0.8\textwidth]{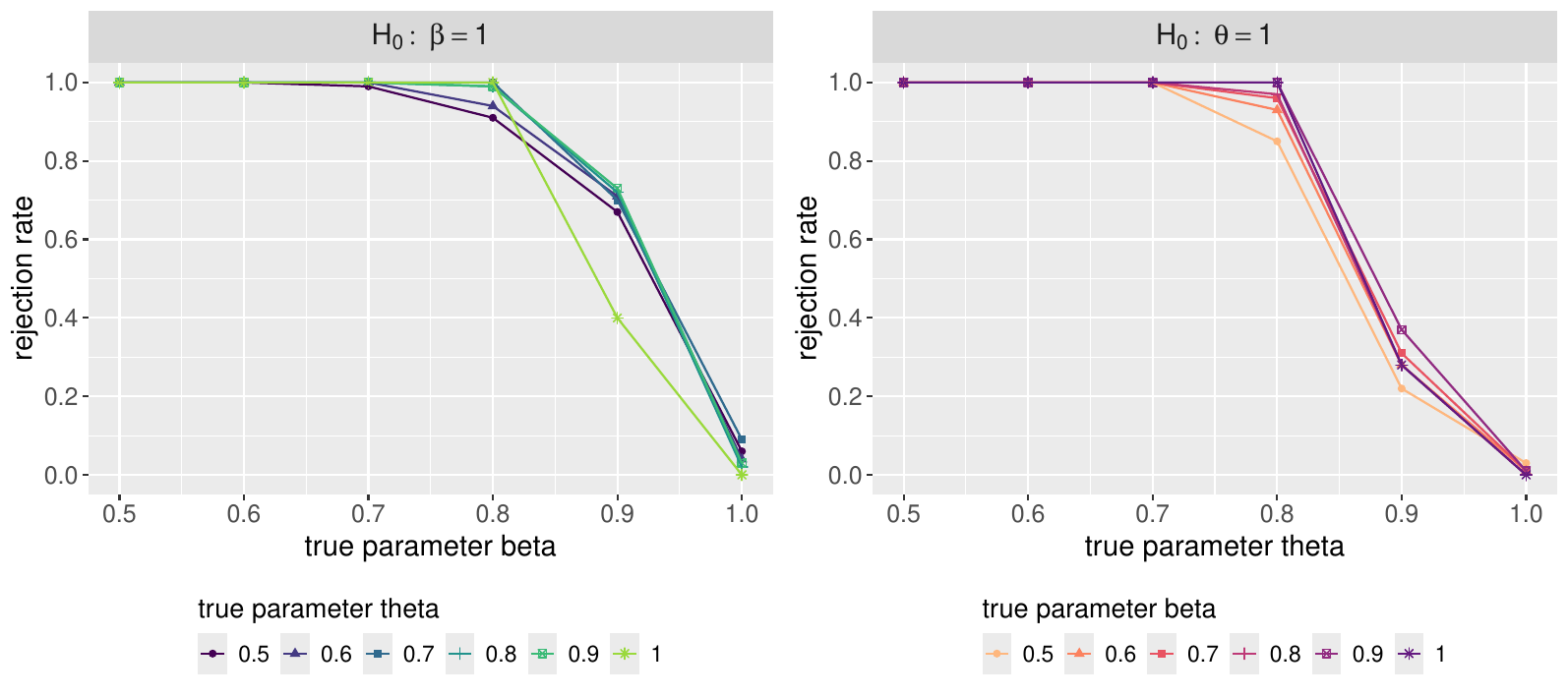}
	\caption{Rejection rates of the bootstrap-tests for the hypothesis $\beta=1$ (left) and $\theta=1$ (right) with $n=10000$ and $k=200$.}
	\label{fig:sim5}
\end{figure}

Since we do not know the true standard errors of our estimators, we estimated them in the data analysis in Section \ref{sec:example} from  bootstrap samples. Figure~\ref{fig:sim6} shows, for three selected parameter combinations similar to those for the three locations in the data analysis, that the bootstrap procedure yields reasonable estimates, with the boxplots representing the distribution of the bootstrap standard errors and the blue dot representing the standard error estimated from $N=100$ simulation runs performed with the true parameter values.

\begin{figure}[htb]
	\centering
	\includegraphics[width = 1\textwidth]{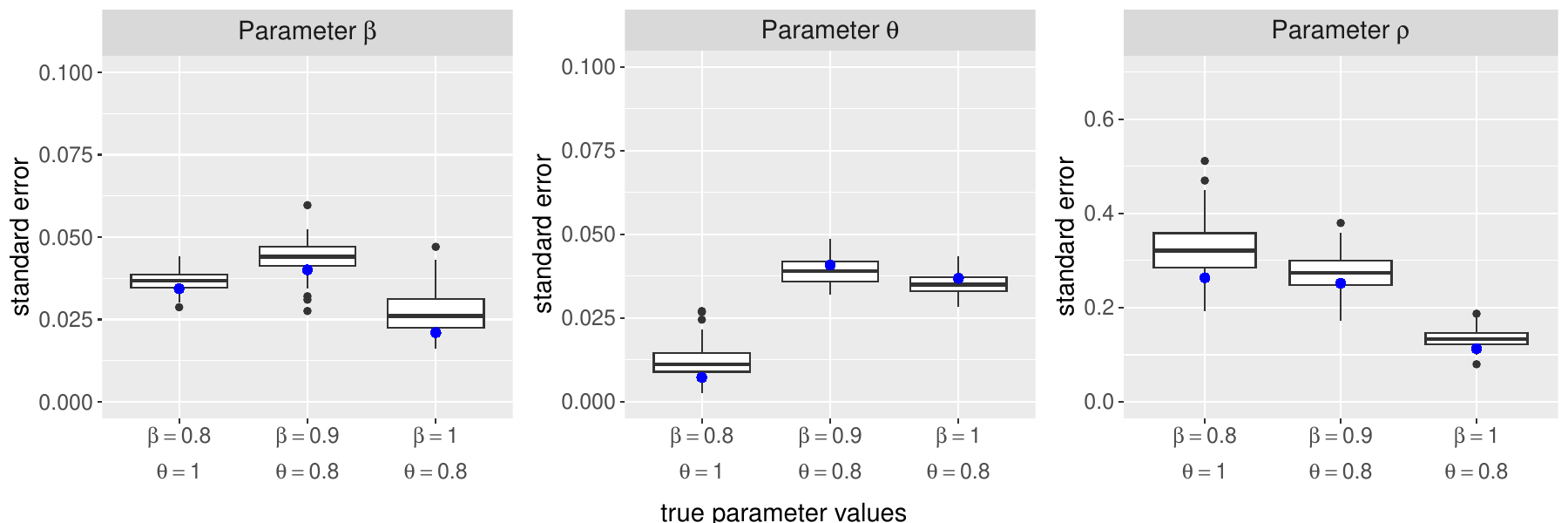}
	\caption{Comparison of the bootstrap standard errors (boxplot) with standard errors estimated by simulations (blue points)   for three scenarios.}
	\label{fig:sim6}
\end{figure}

\section{Performance on alternative clustering processes} \label{sec:simStudy2}

Now we consider two other point processes that also describe temporal clustering as alternatives to the FCPP.
We inspect the respective advantages and limitations and investigate which aspects and behaviours can be captured by the different approaches.

The first alternative considered is the Hawkes process, which is widely used to model temporal clustering in point processes. In a Hawkes process, the occurrence of an event increases the probability of events in the near future, thereby creating temporal clustering. It has been applied across various fields, most prominently in finance (see, e.g., \cite{Hawkes_finance}), but also in areas such as seismology \citep{hawkes_earthquakes} or ecology \citep{hawkes_ecology}. 
Furthermore, the log-Gaussian Cox process (LGCP) is often used to model clustering in point processes, especially in spatial and spatio-temporal settings. It is a Poisson process with its logarithmic intensity being a Gaussian process. Recent application examples include forest fires in Nepal \citep{lgcp_fire} or seismic activity in Greece and Italy \citep{lgcp_earthquakes}.

In the following simulations, we investigate how the FCPP performs when applied to data generated from a Hawkes process or an LGCP. This allows us to examine how other clustering mechanisms are captured by the FCPP.  
To this end, we simulate data from both alternative processes and analyse the parameters and distributions obtained by fitting an FCPP to the simulated data. 
In order to consider scenarios which are of interest in our context, we 
first fit the Hawkes process and the LGCP to the mid-latitude cyclone data (see Section~\ref{sec:example}) observed at Locations~A, B and C, which were already used in Figure~\ref{fig:exm-4} of Section~\ref{sec:cyclones}. Then, we generate 300 datasets with 320 observations from each of the fitted models, similar to our sample sizes.  

For fitting and simulating the Hawkes process we use the R~package \texttt{stelfi} \citep{stelfi}. The Hawkes process considers a baseline rate $\mu_\text{H}$ for the intensity, which is increased by a value $\alpha_\text{H}$ immediately after an event, followed by an exponential decay of this increase controlled by the parameter $\beta_\text{H}$.

Estimation and simulation of the temporal LGCP are implemented using the R~package \texttt{INLA} \citep{INLA}. We denote the mean of the underlying Gaussian process by $\mu_\text{lGC}$, while $\sigma^2_\text{lGC}$ and $V_\text{lGC}$ are the variance and the range parameter of the exponential covariance function.

The parameters obtained from fitting the Hawkes, the LGCP or the FCPP process to the cyclone data at Locations~A, B and C are reported in Table~\ref{tab_competitors}. At Location~A, an FCPP with $\beta=1$ is fitted, corresponding to the special case of a CPP, i.e., clusters of events separated by exponential return times. The Hawkes process adapts to such clustering by a rather large increase of the intensity governed by $\alpha_\text{H}$. The LGCP fits these data with a large variance and a small range parameter. At Location~C, the situation is reversed, as the FCPP fit yields $\theta=1$, i.e., an FPP with heavy-tailed inter-exceedance times. The Hawkes process is not designed for such scenarios and fits a very small intensity increase $\alpha_\text{H}$ after an event with a rapid decay regulated by $\beta_\text{H}$, while the LGCP uses a long range to describe such data. At Location~B the FCPP combines the features of local clustering and heavy-tailed return times. The fitted Hawkes process shows a large increase of the intensity followed by a slower decay compared to the other locations. For the LGCP, the estimated variance and range parameters lie between those for the other locations.

\begin{table}[htb]
	\centering
	\caption{Parameter estimates for the Hawkes process, the LGCP and the FCPP at Locations~A, B and C.}
	\label{tab_competitors}
	\begin{tabular}[h]{r | ccc | ccc | cc}
		\toprule
		 & \multicolumn{3}{c}{Hawkes} &  \multicolumn{3}{c}{LGCP} & \multicolumn{2}{c}{FCPP} \\
		Location & $\mu_\text{H}$ & $\alpha_\text{H}$ & $\beta_\text{H}$ & $\mu_\text{lGC}$ & $\sigma^2_\text{lGC}$ &  $V_\text{lGC}$ & $\theta$ & $\beta$  \\
		\midrule
		A & 0.035 & 0.4118 & \phantom{111}2.21 & -3.93 & 2.01 & \phantom{1}1.15 & 0.83 & 1.00 \\
        B & 0.026 & 0.2791 &  \phantom{111}1.17 & -4.34 & 1.71 & \phantom{1}5.10 & 0.84 & 0.88\\
        C & 0.027 & 0.0002 & 7591.14 & -4.30 & 1.31  & 13.89 & 1.00 & 0.77\\
		\bottomrule
	\end{tabular}
\end{table}

Figure~\ref{fig:sim7} depicts boxplots of the FCPP parameter estimates obtained for the data simulated from the LGCP and the Hawkes processes fitted to Locations~A, B and C. The FCPP parameter estimates obtained for the underlying real mid-latitude cyclone data are represented by points of the same colour.
The results confirm that the Hawkes process has little ability to generate return times with a heavy upper tail. 
This is particularly evident at Location~C, where the estimated tail parameter $\beta$ is 0.77 for the original data, while all estimates obtained from the simulated Hawkes data exceed 0.94.
At Location B, the estimates of $\theta$ are smaller on average than for the real data, whereas the estimates of $\beta$ are generally larger, giving more weight to the lower and less to the upper tail. This again suggests that the Hawkes process may not adequately reproduce the relatively heavy upper tail found for the original data.
For Location~A, the estimates obtained from the simulated Hawkes process are similar to those for the real data, since there is no indication of a heavy upper tail in the return times of the original data.

\begin{figure}[htb] 
	\centering
	\includegraphics[width = 0.8\textwidth]{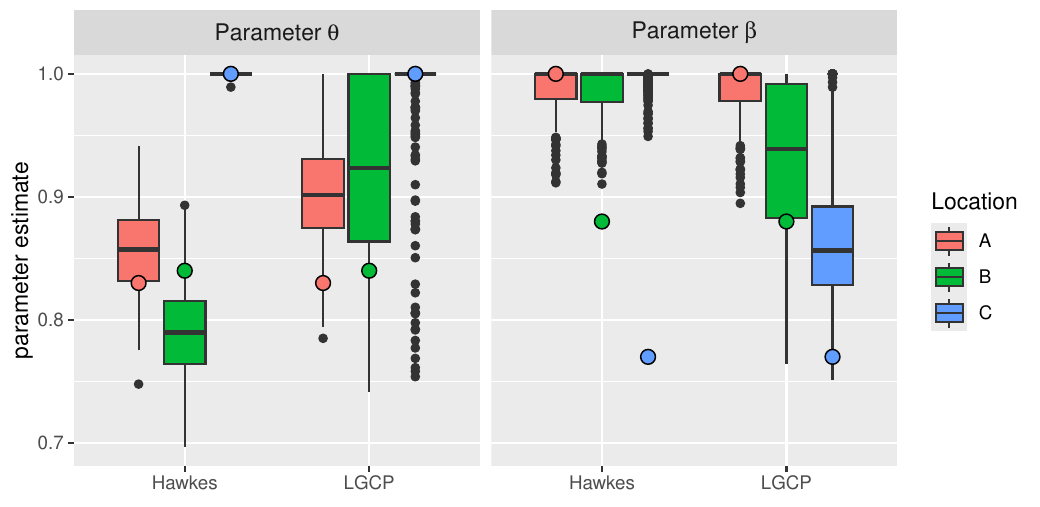}
	\caption{Boxplots of the FCPP parameter estimates $\hat{\theta}$ and $\hat{\beta}$ when fit to data simulated from Hawkes and log-Gaussian Cox processes. The parameter estimates obtained from the original data are shown by points of the same colour.}
	\label{fig:sim7}
\end{figure}

The results obtained for the LGCP data differ notably from those for the Hawkes process. For $\theta$, the estimates based on the data simulated from a LGCP at Locations~A and B tend to be slightly higher than the corresponding estimates for the original data. Similarly, most estimates of $\beta$ from the simulated LGCP data at Locations~B and C are higher than the original estimates. 
In the special cases of $\hat{\theta} = 1 $ at Location~C and $\hat{\beta} = 1 $ at Location~A, the median estimate from the simulated LGCP data coincides with the original estimate.
Overall, these findings indicate that the LGCP captures and generates both heavy lower and upper tail behaviour to some extent, but less so than the FCPP. 

We also evaluate how closely the return time distributions of the fitted FCPPs in Figure~\ref{fig:sim7} resemble those of the underlying LGCP or Hawkes process. Since we do not have an analytic formula available for the latter, we compute the two-sample Cramér–von Mises distances \citep{twosample_cvm} between datasets simulated from the competitor processes and datasets generated from the FCPP fitted to them. Specifically, we generate 150 datasets from every fitted FCPP and compute the two-sample Cramér-von Mises distance between the datasets simulated from the FCPP and those simulated from the corresponding LGCP or Hawkes process. These distances are then compared with the two-sample Cramér-von Mises distances between datasets simulated from the underlying LGCP or the Hawkes process. 
This comparison allows to assess how strongly the FCPP deviates from the Hawkes process or LGCP relative to the variability among datasets drawn from this Hawkes process or LGCP. The results are presented in Figure~\ref{fig:sim8}. 

\begin{figure}[htb] 
	\centering
	\includegraphics[width = 0.8\textwidth]{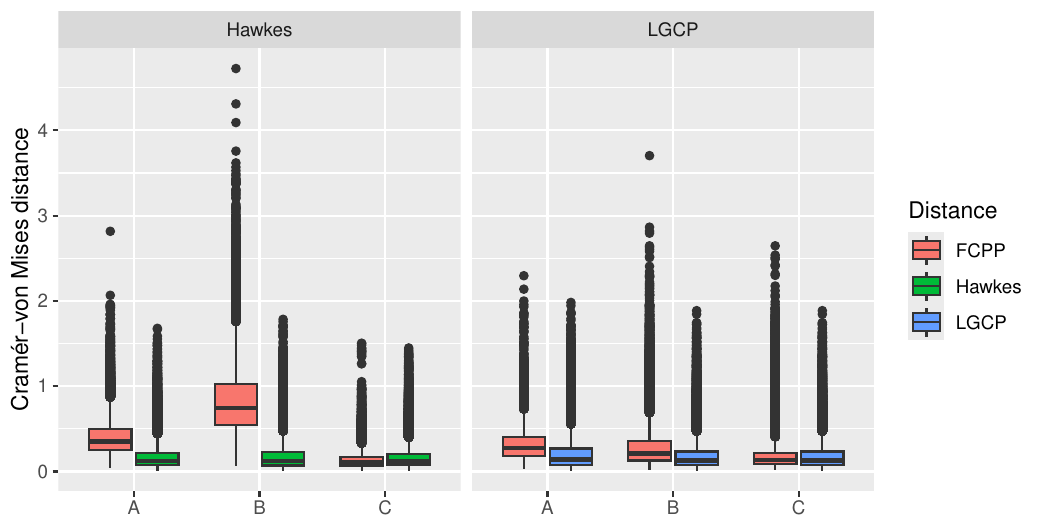}
	\caption{Boxplots of the two-sample Cramér-von Mises distances between different datasets simulated from the same Hawkes process or the same LGCP (with parameters corresponding to Location~A, B or C) and between such datasets and datasets generated from the FCPP fitted to them.}
	\label{fig:sim8}
\end{figure}

The two-sample Cramér-von Mises distances between the FCPP and LGCP datasets are only slightly larger than the distances among the LGCP datasets themselves, indicating a good fit of the FCPP model to the LGCP. At Location~C, the FCPP seems to capture the LGCP particularly well, but the results look adequate at all three locations.

The distances between the FCPP and the Hawkes process datasets vary more strongly across the three locations. 
At Location~C, the FCPP seems to capture and reproduce the return times of the underlying Hawkes process well. 
At Location~A, the Cramér-von Mises distances between the FCPP and the Hawkes process datasets are somewhat larger than the distances between the Hawkes process datasets themselves. 
At Location~B, the FCPP datasets differ more strongly from the Hawkes process datasets. At this location, the Hawkes process fits a relatively small value of $\beta_\text{H}$ corresponding to a long-lasting increase of the intensity caused by past events, which is not captured very well by the asymptotic point mass in 0 used by the FCPP. Moreover, the FCPP sometimes fits return times with a heavy upper tail, albeit the underlying Hawkes process does not seem to generate return times with this feature. 

In conclusion, this study indicates some problems of the FCPP fitting return times with a heavy lower tail which is spread out over a wide range, as in case of a Hawkes process with a small value of $\beta_\text{H}$. In contrast, Hawkes processes show difficulties in fitting return times with heavy upper tails. 
As opposed to this, the FCPP fitted data generated from a LGCP well in our study, and in turn the LGCP is able to cover heavy upper and lower tail behaviour to some extent. An advantage of the FCPP in the context of extreme value statistics is that it arises from sound asymptotical reasoning and allows extrapolation to IETs between even more extreme event magnitudes.

\section{Data analysis on mid-latitude cyclones}
\label{sec:example}


Now we continue our analysis of the occurrences of extreme mid-latitude cyclones on the northern hemisphere. 
We apply the fractional compound  Poisson process (FCPP) introduced in Section \ref{sec:probabilisticmodel} to the occurrences of mid-latitude cyclones and compare it with its special cases PP, FPP and CPP. 
This allows us to determine whether the IETs can be better described by the exponential, or by the Mittag-Leffler distribution, or by a mixture distribution with an exponential component, or by a combination of both. For a better overview, Table \ref{tab:2} summarizes the four models with the corresponding IET distributions and unknown parameters. 

\begin{table}[ht]
\begin{tabular}{lrr}
\hline
\textbf{model} & \textbf{IET distribution} & \textbf{unknown parameters} \\
\hline
PP             & exponential               & $\sigma$                    \\
FPP            & Mittag-Leffler            & $\beta$, $\sigma$           \\
CPP            & mixed exponential         & $\theta$, $\sigma$          \\
FCPP           & mixed Mittag-Leffler      & $\beta$, $\theta$, $\sigma$   
\end{tabular}
\caption{Summary of the four IET models.}
\label{tab:2}
\end{table}

We use the same data source and method for identifying extreme cyclones as \citet{Blender2015}. There are some differences, as data for a longer time period starting in 1940 with a higher horizontal resolution are available now.
In addition, we classify all locations using a parametric bootstrap. For this, we generate B = 100 independent bootstrap samples for each location from the data-generating process fitted to the real data (for more details see the previous section on simulations), and we re-estimate the parameters using CMmod. If the tail parameter $\beta$ is estimated to be smaller than one in 95\% or more of the samples, then we assume $\beta < 1$ to be true; similarly, we assume $\theta < 1$ to be true if 95\% of the sample estimates $\theta$ are smaller than one.

Our results focus on the two parameters $\beta$ and $\theta$ since they capture the clustering behaviour. All three models (see Figure \ref{fig:exm-1}) fit well to the general pattern that serial clustering occurs at the exit region of storm tracks to the west of Europe, while they occur more regularly at the American east coast (\cite{Dacre2020}). Moreover, the mountains in southern Europe seem to have a large influence on the IET distribution, as we find the strongest clustering behaviour regarding both parameters $\beta$ and $\theta$ there. 

\begin{figure}[!ht]
    \centering
    \includegraphics[width=1\textwidth]{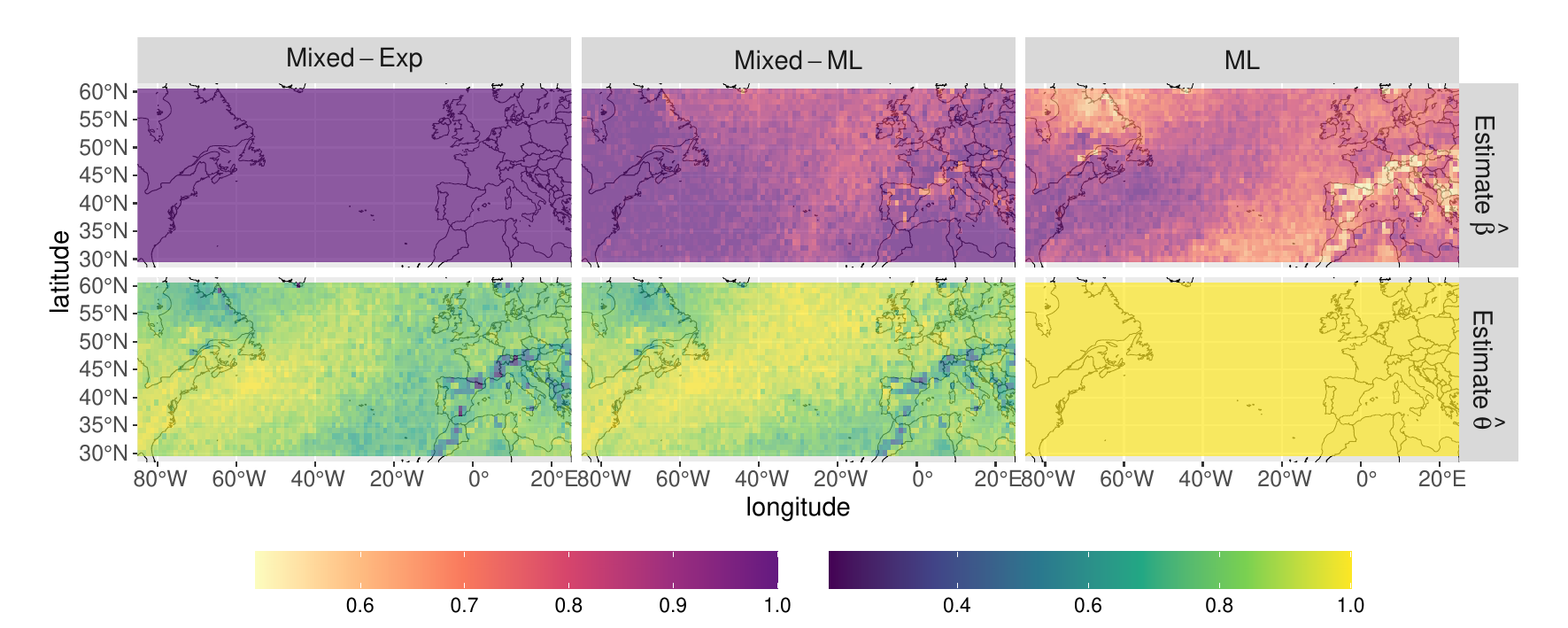}
    \caption{Estimations of tail parameter $\beta$ and extremal index $\theta$ in case of the CPP (left), the FCPP (middle) and FPP (right). The tail parameter is equal to 1 in the CPP, while this applies to $\theta$ in the FPP.}
    \label{fig:exm-1}
\end{figure}
Comparing the results for the FCPP and the FPP model (see Figure \ref{fig:exm-1}, middle and right column), we see that the tail parameter $\beta$ is generally estimated larger in the FCPP. Except for the storm track exit region over the north Atlantic and European mountain areas the tail parameter $\beta$ is estimated mostly close to one, while the FPP model suggests lower values of $\beta$. The reason for the difference between the results for the FCPP and FPP model is that the FCPP is more flexible and explains the serial clustering via both effects, the mixture component and heavy tails, so that a larger value of $\beta$ is compensated in the FCPP model by a extremal index $\theta$ less than 1 in these regions.

\begin{figure}[!t]
    \centering
    \includegraphics[width=0.5\textwidth]{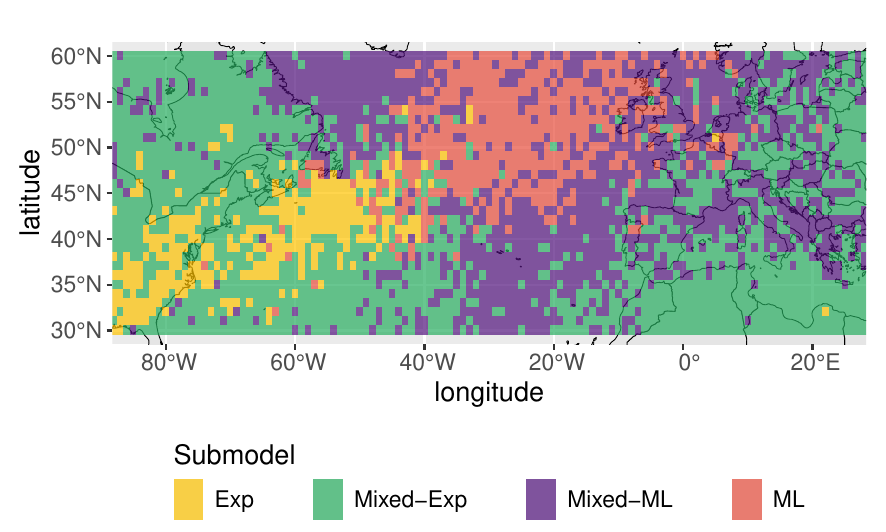}
    \caption{Bootstrap classification of extreme mid-latitude cyclones into the four (sub)models.}
    \label{fig:exm-3}
\end{figure}

When comparing the results of the FCPP model to the CPP model (see Figure \ref{fig:exm-1}, middle and left column), we observe fewer differences concerning the extremal index $\theta$. This is due to the tail parameter $\beta$ being estimated to be close to one at most locations, which puts us in the special case of the CPP. In regions where $\beta$ is estimated to be clearly less than one, the estimate of $\theta$ is higher than in the CPP model.
The results of the bootstrap classification (see Figure \ref{fig:exm-3}) underline these results.

To conclude this analysis, the three illustrative examples considered in Section \ref{sec:cyclones} are analysed further. The three locations have been selected in such a manner that Location A corresponds to the CPP submodel, Location C to the FPP submodel, and Location B to the general FCPP model according to the bootstrap classification.
One arguable assumption for all considered models is the underlying uncoupled marked point process, i.e., the stochastical independence of the event magnitudes and waiting times. We consider the empirical copula plot (see Figure \ref{fig:exm-7}) as diagnostic tool to investigate the dependence between the exceedances $X_i(u)$ and the IETs $T_i(u)$ as proposed in \citet{Hees2021}. The graph gives no indication of a possible dependence. Instead, it shows that for Location A and B, for both of which $\theta < 1$ was estimated, there is indeed an accumulation of consecutive exceedances.

\begin{figure}[!h]
    \centering
    \includegraphics[width=1\textwidth]{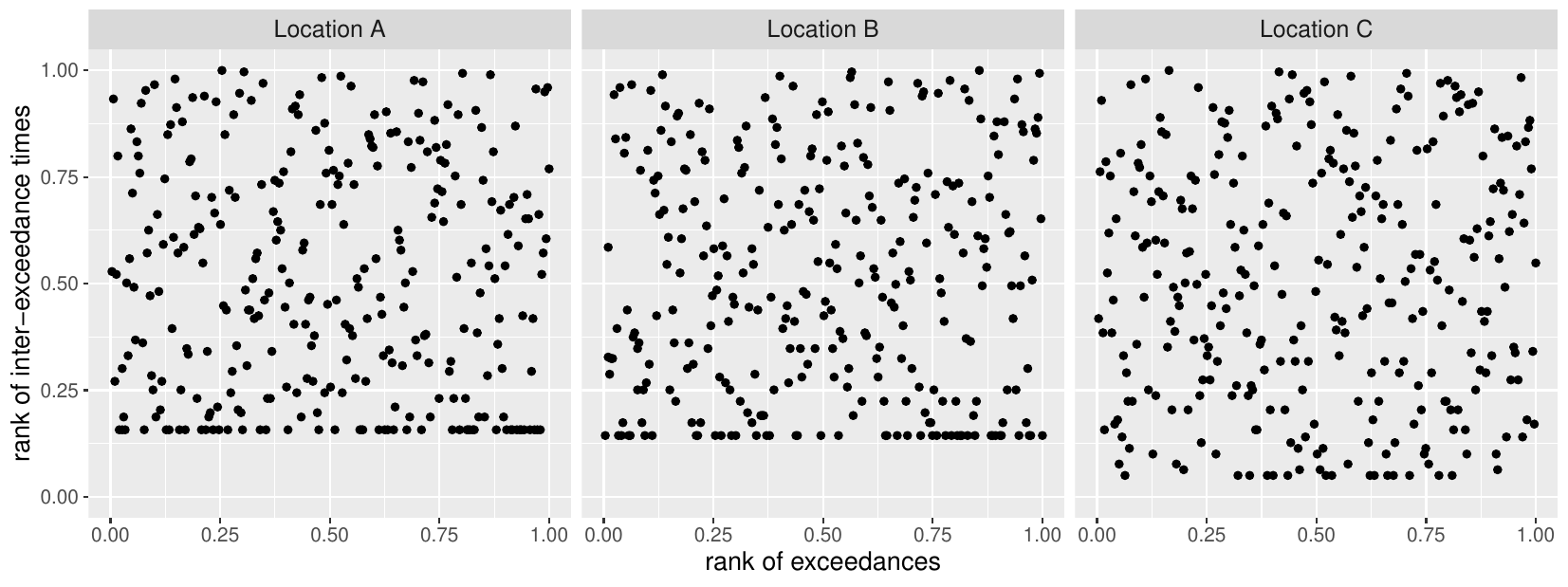}
    \caption{Empirical copula plot for the exceedances and the IETs.}
    \label{fig:exm-7}
\end{figure}

Table \ref{tab:1} shows the parameter estimations of the CMmod approach for Locations A, B and C along with estimated standard errors by a parametric bootstrap procedure as proposed in the last paragraph of Section \ref{sec:simStudy}. For the CPP and FPP we assume $\beta = 1$ or $\theta = 1$ to be known true, respectively, and optimise only about the remaining parameters.

\begin{table}[htb]
	\centering
	\caption{Parameter estimations for the three Locations A, B and C using CMmod. Marked in grey are the submodels that are equal to the estimation values of the FCPP. Second rows report the estimated standard errors by a parametric bootstrap procedure.}
	\label{tab:1}
	\begin{tabular}[h]{r | cc | cc | cc}
		\toprule
		& \multicolumn{2}{c}{Location A} &  \multicolumn{2}{c}{Location B} &  \multicolumn{2}{c}{Location C} \\
		& $\beta$ & $\theta$ &  $\beta$ & $\theta$ &  $\beta$ & $\theta$ \\
		\midrule
		FCPP & 1.00 & 0.83 & 0.88 & 0.84 & 0.77 & 1.00 \\
		& (.020) & (.031) & (.035) & (.030) & (.025) & (.010)  \\
		CPP & {\color{gray}1.00} & {\color{gray}0.83} & 1.00 & 0.76 & 1.00 & 0.74 \\
		& & (.033) & & (.034) & & (.032) \\
		FPP & 0.81 & 1.00 & 0.72 & 1.00 & {\color{gray}0.77} & {\color{gray}1.00} \\
		& (.027)  &  & (.028) & & (.026) &\\
		\bottomrule
	\end{tabular}
\end{table}

Figure \ref{fig:exm-6}, top row, shows histograms of the three locations with a bar width of one day. The
fitted densities of the three models CPP, FCPP and FPP are included for comparison.
We observe differences between the locations concerning the parameter estimation of $\beta$ and $\theta$. At Locations B and C $\beta$ is estimated to be less than one. The right tail of the distribution is apparently heavier there than at Location A, where $\beta$ is estimated to be equal to one. At Locations A and B, $\theta$ is estimated to be less than one. This is due to the high number of very small IETs that do not exceed one day.
The middle row of Figure \ref{fig:exm-6} shows QQ plots. Especially the upper quantiles of Location C, adjusted for FCPP, seem to fit worse than those of CPP. To answer the question of why FCPP did not estimate $\beta = 1$ and $\theta < 1$ in this case, we examine the same QQ plots and focus on the small IETs (see Figure \ref{fig:exm-6}, bottom row). Here, we can see that FCPP describes the empirical quantiles at all three locations best. It is noteworthy that over 90\% of the IETs at all three locations are shorter than 100 days. Thus, the CMmod estimation method gives more weight to smaller observations in the estimation. 
\begin{figure}[!t]
    \centering
    \includegraphics[width=1\textwidth]{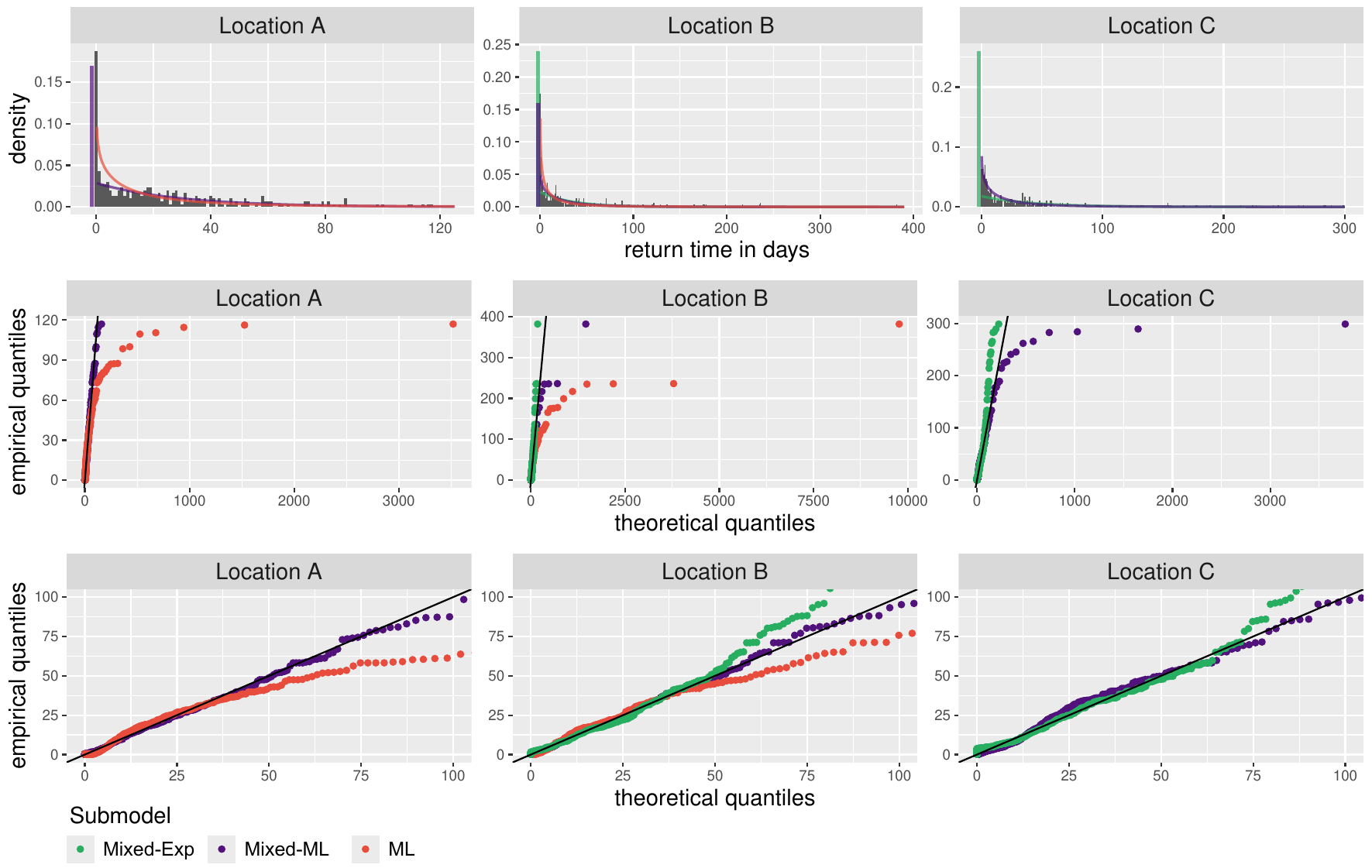}
    \caption{Histogram (top) and QQ plots (middle, bottom) of the IETs of Locations A, B and C with densities and theoretical quantiles, respectively, fitted using the mixed exponential, Mittag-Leffler and mixed Mittag-Leffler distribution. At the bottom row the QQ plot focuses on small IETs up to 100 days.}
    \label{fig:exm-6}
\end{figure}

Lastly, we examine the impact of the various models on the probabilities associated with the IETs using Location B as illustrative example.
Table \ref{tab:datasample2} presents the probabilities of the IETs not exceeding 1, 2, 7, 30, 100, and 365 days, respectively.

\vspace{1.5em}
\begin{table}[htb]
	\centering
	\caption{Estimated probabilities (in percent) of IETs not exceeding 1, 2, 7, 30, 100, and 365 days at Location B for the four (sub)models.}
	\label{tab:datasample2}
	\begin{tabular}[h]{r | rrrrrr}
		\toprule
		t (days) & 1 & 2 & 7 & 30 & 100 & 365  \\
		\midrule
		FCPP & 17.67 & 18.80 & 23.60 & 39.48 & 65.86 & 91.94 \\
		CPP & 24.58 & 25.15 & 27.96 & 39.58 & 64.62 & 95.34 \\
		FPP & 4.86 &  7.85 & 17.96 & 41.32 & 67.38 & 87.78 \\
		PP & 1.23 &  2.45 &  8.32 & 31.08 & 71.08 & 98.92 \\
		\bottomrule
	\end{tabular}
\end{table}
\vspace{1em}

Since we do not know the truth, we cannot conclusively determine which model is nearest to it. However, we observe large differences between the estimates, particularly for the very short times. For intermediate IETs (100 days), the estimates from the four models are similar, but they differ again for even longer times. The compound models, CPP and FCPP, in particular predict a high probability, around 19\% and 25\% respectively, of another extreme event in the next two days. In contrast, the FPP estimates this probability to be below 8\%. The probability of another event in the next 7 days is estimated to be about 18\% or larger in the three models which take clustering into account.

In general, our findings provide evidence supporting the hypothesis of \citet{Blender2015} that the deviation of the dispersion from zero in the exit region of the storm tracks can be described, at least in part, by the Mittag-Leffler distribution. Beyond that, our results indicate that using a mixture model often provides even better fits. The tail parameter might be estimated too small otherwise, pretending a too heavy distributional tail.

\section{Conclusion}
\label{sec:conclusion}

Extremes above a high threshold often occur in temporal clusters, i.e., several extreme values occur in a short period of time, followed by a longer period without such extremes. The inter-exceedance times (IETs) are then poorly described by an exponential distribution derived from a Poisson process. There are several asymptotical modelling approaches to capture deviation from exponential return times. One of them is the compound Poisson process (CPP) which corresponds to a mixture distribution of the Dirac measure at zero and an exponential distribution (\cite{Ferro2003}). Another model is the fractional Poisson process (FPP) where the IETs are asymptotically Mittag-Leffler distributed (\cite{Hees2021}) with tails heavier than those of the exponential distribution. 
In the present work we have combined these two approaches. Relaxing the conditions for the classical Poisson process results into both directions, we consider events that are stationary and  separated by heavy-tailed waiting times. Asymptotically the IETs then follow a fractional compound Poisson process (FCPP), which corresponds to a mixture distribution of the Dirac measure at zero and a Mittag-Leffler distribution. This model has three parameters, namely the tail parameter $\beta$, the extremal index $\theta$ and the scaling parameter $\sigma_{p(u)}$. The CPP and the FPP correspond to the special cases  $\beta = 1$ and $\theta = 1$, respectively. 
 
For estimating these three parameters we propose CMmod, a minimum distance approach based on a modification of the Cramér-von Mises distance. Our simulation study illustrates the suitability of the CMmod estimation, although the bias and RMSE are slightly higher in case of a low extremal index $\theta$ and Mittag-Leffler and exponentially distributed waiting times compared to the other scenarios we considered. 
In the special cases $\beta = 1$ and $\theta = 1$ it performs competitively and sometimes even better than common estimation methods designed for these scenarios.  In our simulations and real data analysis a parameter which was not needed for describing the data was often estimated to be equal or close to 1.
Thus, there seems to be little disadvantage when fitting the more general FCPP model, except for the longer computing time.
We thus do not need to decide in advance which model for clustering provides the better description of the data.

In our application to mid-latitude winter cyclones we have illustrated that the IETs occur in clusters and are poorly described by the exponential distribution. We have seen that different (sub)models  provide the best fit to the data depending on the exact location (off-shore in the Atlantic, west shore of Europe, interior of the continent, in the mountains, etc.). Simulations confirm that the proposed bootstrap tests offer a conservative approach for classifying whether fitting the FCPP model is required or if a submodel suffices for the data.
This study has not addressed the issue of seasonality and trends, such as those potentially driven by climate change, although increasing storm intensity could be expected as global temperatures rise (\cite{Karwat2022}).

\section*{Acknowledgements}
This research has been funded by the German Federal Ministry of Education and Research (BMBF) within the subproject SCAHA (project number 01LP1902K) of the research network on climate change and extreme events (climXtreme). The authors gratefully acknowledge the computing time provided on the Linux HPC cluster at TU Dortmund University (LiDO3), partially funded in the course of the Large-Scale Equipment Initiative by the German Research Foundation (DFG) as project 271512359. The authors would like to thank Prof. Richard Blender and Dr. Alexia Karwat for helpful discussions and insights on meteorological topics, especially mid-latitude winter cyclones.

\section*{Data availability statement}
The R-Code used for the simulation study will be provided as supplement. The ERA5 reanalysis data are openly available in the Climate Data Store of the ECMWF at \url{https://cds.climate.copernicus.eu/cdsapp/#!/search?type=dataset}.

\emergencystretch=1em 
\setlength\bibitemsep{5pt}
\printbibliography

\end{document}